\documentclass[a4paper,USenglish]{lipics-v2016}
 
\usepackage{microtype}
\usepackage{amsmath,amssymb,amsfonts,amsthm}
\usepackage{algorithm,algpseudocode, booktabs}
\usepackage[usenames,dvipsnames]{xcolor}

\newcommand{\vc}[1]{\boldsymbol{#1}}

\newcommand{\eps}{\varepsilon}

\renewcommand{\Pr}{\mathbb{P}}
\newcommand{\expn}{\mathbb{E}}
\newcommand{\cD}{\mathcal{D}}

\algnotext{EndFor}
\algnotext{EndIf}
\algnotext{EndWhile}

\newtheorem{prop}[theorem]{Proposition}

\newcommand{\rhoq}{\rho_{\mathrm{q}}}
\newcommand{\rhos}{\rho_{\mathrm{s}}}

\newcommand{\cS}{\mathcal{S}}
\newcommand{\C}{\mathcal{C}}
\newcommand{\W}{\mathcal{W}}
\newcommand{\cL}{\mathcal{L}}

\newcommand{\U}{\mathcal{U}}
\newcommand{\R}{\mathbb{R}}
\newcommand{\B}{\mathcal{B}}
\newcommand{\Z}{\mathbb{Z}}
\newcommand{\Var}{\mathrm{Var}}
\newcommand{\FALCONN}{\texttt{FALCONN}}
\newcommand{\NMSLib}{\texttt{NMSLib}}
\newcommand{\KGraph}{\texttt{KGraph}}

\newcommand{\poly}{\operatorname{poly}}
\newcommand{\polylog}{\operatorname{polylog}}

\newcommand{\ip}[2]{\langle #1, #2 \rangle}

\theoremstyle{plain}

\title{Graph-based time--space trade-offs for approximate near neighbors}
\titlerunning{Graph-based time--space trade-offs for approximate near neighbors} 

\author[1]{Thijs Laarhoven}
\affil[1]{Eindhoven University of Technology\\Eindhoven, The Netherlands\\ \texttt{mail@thijs.com}}
\authorrunning{Thijs Laarhoven}
\Copyright{Thijs Laarhoven}


\EventEditors{TBD}
\EventNoEds{1}
\EventLongTitle{TBD}
\EventShortTitle{TBD}
\EventAcronym{TBD}
\EventYear{2018}
\EventDate{TBD}
\EventLocation{TBD}
\EventLogo{}
\SeriesVolume{1}
\ArticleNo{1} 

\begin{document}

\maketitle

\begin{abstract}
We take a first step towards a rigorous asymptotic analysis of graph-based approaches for finding (approximate) nearest neighbors in high-dimensional spaces, by analyzing the complexity of (randomized) greedy walks on the approximate near neighbor graph. For random data sets of size $n = 2^{o(d)}$ on the $d$-dimensional Euclidean unit sphere, using near neighbor graphs we can provably solve the approximate nearest neighbor problem with approximation factor $c > 1$ in query time $n^{\rhoq + o(1)}$ and space $n^{1 + \rhos + o(1)}$, for arbitrary $\rhoq, \rhos \geq 0$ satisfying
\begin{align}
(2c^2 - 1) \rhoq + 2 c^2 (c^2 - 1) \sqrt{\rhos (1 - \rhos)} \geq c^4.
\end{align}
Graph-based near neighbor searching is especially competitive with hash-based methods for small $c$ and near-linear memory, and in this regime the asymptotic scaling of a greedy graph-based search matches the recent optimal hash-based trade-offs of Andoni--Laarhoven--Razenshteyn--Waingarten~\cite{andoni17}. We further study how the trade-offs scale when the data set is of size $n = 2^{\Theta(d)}$, and analyze asymptotic complexities when applying these results to lattice sieving.
\end{abstract}

\section{Introduction}

\subparagraph{Nearest neighbor searching.} A key computational problem in various areas of research, such as machine learning, pattern recognition, data compression, and decoding~\cite{bishop06, dubiner10, duda00, may15, laarhoven15crypto, shakhnarovich05}, is the \textit{nearest neighbor problem}: given a $d$-dimensional data set $\cD \subset \R^d$ of cardinality $n$, design a data structure and preprocess $\cD$ in an efficient way such that, when later given a query vector $\vc{q} \in \R^d$, we can quickly find the nearest point to $\vc{q}$ in $\cD$ (e.g.\ under the Euclidean metric). Since the exact (worst-case) version of this problem suffers from the ``curse of dimensionality''~\cite{indyk98}, a common relaxation of this problem is the \textit{approximate nearest neighbor (ANN) problem}: given that the exact nearest neighbor in the data set $\cD$ lies at distance at most $r$ from $\vc{q}$, design an efficient algorithm that finds an element $\vc{p} \in \cD$ at distance at most $c \cdot r$ from $\vc{q}$, for a given approximation factor $c > 1$. We will refer to this problem as $(c, r)$-ANN. Since a naive linear search trivially leads to an $O(d \cdot n)$ time algorithm for solving both the exact and approximate near neighbor problems, the goal is to design a data structure for which queries can be accurately answered in time $n^{\rho + o(1)}$ with $\rho < 1$. Here we will only consider scenarios where $d$ scales with $n$ -- for fixed $d$, it is well-known that one can achieve arbitrarily small query exponents $\rho = o(1)$~\cite{arya94}. We further assume w.l.o.g.\ that $n \geq 2^{O(d / \log d)}$ -- in case $d$ is larger, we can first perform a random projection to reduce the effective dimensionality of the data set, while maintaining inter-point distances~\cite{johnson84}.

\subparagraph{Partitioning the space.} A celebrated technique for efficiently solving the ANN problem in high dimensions is \textit{locality-sensitive hashing (LSH)}~\cite{indyk98, datar04, andoni06, andoni14}. Using hash functions with the property that nearby vectors are more likely to be mapped to the same hash value, one builds several hash tables with buckets containing vectors with the same hash value. Queries $\vc{q}$ are then processed by computing $h(\vc{q})$, looking up vectors $\vc{p} \in \cD$ with the same hash value $h(\vc{p}) = h(\vc{q})$, and considering these vectors as candidate near neighbors, for each of the precomputed hash tables. Although a large number of randomized hash tables is often required to obtain a good recall rate, doing these look-ups in all the hash tables is often still considerably faster than a linear search through the list.

Whereas LSH requires each point to be mapped to exactly one hash bucket in each hash table, the recent \textit{locality-sensitive filtering (LSF)}~\cite{becker16lsf, andoni17, christiani17} relaxes this condition: for each hash filter and each point in the data set, we independently decide whether we add this point to this filter bucket or not. This may mean that some points are added to more filters than others, and filters do not necessarily form a partitioning of the space. Queries are answered by considering filters matching the query, and going through all vectors in these buckets.

For the special case of the approximate nearest neighbor problem where the data set lies on the sphere, many efficient partition-based methods are known based on using random hyperplanes~\cite{charikar02}, regular polytopes~\cite{terasawa07, andoni15cp, kennedy17, laarhoven17hypercube}, and spherical caps~\cite{andoni06, andoni14, becker16lsf, laarhoven15nns, andoni17}. For random data sets of size $n = 2^{o(d)}$, both cross-polytope LSH~\cite{andoni15cp} and spherical LSF~\cite{becker16lsf} achieve the optimal asymptotic scaling of the query complexity for hash-based methods. Spherical LSF further offers optimal time--space trade-offs between the query and update/space complexities~\cite{andoni17}, while cross-polytope LSH seems to perform better in practice~\cite{andoni15cp, falconn, annbench2}. No optimality results are known for data sets of size $n = 2^{\Theta(d)}$, but spherical LSF provably outperforms cross-polytope LSH in some regimes~\cite{becker16lsf}. 

\subparagraph{Nearest neighbor graphs.} A different approach to the near neighbor problem involves constructing \textit{nearest neighbor graphs}, where vertices correspond to points $\vc{p} \in \cD$ of the data set, and an edge between two vertices indicates that these points are approximate near neighbors~\cite{brito97, eppstein97, miller97, plaku07, chen09x, connor10, dong11, hajebi11, wang12}. Given a query $\vc{q}$, one starts at an arbitrary node $\vc{p} \in \cD$, and repeatedly attempts to find vertices $\vc{p}'$, connected to $\vc{p}$ in the graph through an edge, which are closer to $\vc{q}$ than the current candidate near neighbor $\vc{p}$. When this process terminates, the resulting vector is either a local or a global minimum, i.e.\ a false positive or the true nearest neighbor to $\vc{q}$. If the graph is sufficiently well-connected, we may hope to solve the (A)NN problem with high probability in one iteration -- otherwise we might start over with a new random $\vc{p} \in \cD$, and hope for success in a reasonable number of attempts. 

Compared to LSH and LSF, which often require storing quite some auxiliary data describing the hash tables/filters, and for which the cost of computing hashes or finding appropriate filters is commonly non-negligible, the nearest neighbor graph approach has much less overhead: for each vector $\vc{p} \in \cD$, one only has to store its nearest neighbors in a list, and look-ups require no additional computations besides comparisons between the query $\vc{q}$ and data points $\vc{p} \in \cD$ in these lists to find nearer neighbors. In practice, graph-based approaches may therefore be more efficient than hash-based methods even if asymptotic analyses suggest otherwise, purely due to the low overhead of graph-based methods.

Besides the standard nearest neighbor graph approach of connecting each vertex to its nearest neighbors~\cite{dong11, kgraph}, various heuristic graph-based methods are further known based on adding connections in the graph between points which are not necessaily near neighbors~\cite{ponomarenko11, malkov14, nmslib}. This might for instance involve including several long-range, deep links in the graph between distant points, to guarantee that every pair of nodes is connected through a short path~\cite{malkov14}. Using hierarchical graphs~\cite{malkov16, nmslib} where vertices are partitioned in several layers further seems to aid the performance of graph-based methods, although again these improvements appear to be purely heuristic.

\subparagraph{Comparison of different methods.} To find out which method for finding nearest neighbors in high-dimensional spaces is objectively the best, an obvious approach would be to implement these methods, test them against realistic data sets, and compare their concrete performance. Various libraries with implementations of near neighbor methods are openly available online~\cite{annoy, falconn, flann, kgraph, nmslib}, and recently Aum\"{u}ller--Bernhardsson--Faithfull presented a thorough comparison of different methods for nearest neighbor searching on various commonly used data sets and distance metrics~\cite{annbench1, annbench2}. Their final conclusions include that the LSH-based \FALCONN{} and the graph-based \NMSLib{} and \KGraph{} currently seem to be the most competitive methods for their considered data sets. 

Practical comparisons clearly have various drawbacks, as the practical performance often depends on many additional variables, such as specific properties of the data set and the level of optimization of the implementation -- better results on certain data sets may not always translate well to other data sets with further optimized implementations. Moreover, some methods may scale better as the dimensionality and size of the data set increases, and it is hard to accurately predict asymptotic behavior from practical experiments.

A second approach for comparing different methods would be to look at their theoretical, proven asymptotic performance as the size of the data set $n$ and the dimension $d$ tend to infinity. Tight bounds on the performance would allow us to extrapolate to arbitrary data sets, but for many algorithms obtaining such tight asymptotic bounds appears challenging. Although for partition-based methods, by now the asymptotic performance seems to be reasonably well understood, graph-based algorithms are mostly based on unproven heuristics, and are not as well understood theoretically. This is particularly disconcerting due to the efficiency of certain graph-based approaches, and so a natural question is: how well do graph-based approaches hold up against partition-based methods in high dimensions? Is it just the low overhead of graph-based methods that allows them to compete with hash-based approaches? Or will these graph-based methods remain competitive even for huge data sets?

\subparagraph{Heuristic methods made theoretical.} We further remark that in the past, theoretical analyses of heuristic near neighbor methods have not only contributed to a better understanding of these methods, but also to make these methods more practical. Cross-polytope LSH, which is used in \FALCONN{}~\cite{falconn}, was originally proposed as a fast heuristic method with no proven asymptotic guarantees~\cite{terasawa07}; only later was it discovered that this method is theoretically superior to other methods as well~\cite{andoni15cp, kennedy17}, and can be made even more practical. Recently also for hypercube LSH~\cite{terasawa07} improved theoretical guarantees were obtained~\cite{laarhoven17hypercube}, and the heuristic LSH forest approach~\cite{bawa05} was also ``made theoretical'' with provable performance guarantees in high dimensions~\cite{andoni17b}. The goals of a theoretical analysis of graph-based methods are therefore twofold: to understand their asymptotic performance better, and to find ways to further improve these methods in practice.

\subsection{Contributions}

We take a first step towards a better theoretical understanding of graph-based approaches for the (approximate) nearest neighbor problem, by rigorously analyzing the asymptotic performance of the basic greedy nearest neighbor graph approach. We show that for random data sets on the Euclidean unit sphere and for arbitrary approximation factors $c > 1$, this method provably achieves asymptotic query exponents $\rho < 1$. We further show how to obtain efficient time--space trade-offs, by making the related near neighbor graph either more connected (more space, better query time) or less connected (less space, worse query time).  

\subparagraph{Sparse data sets.} In the case the data set on the unit sphere\footnote{Note that the near neighbor problem on the Euclidean unit sphere is of particular interest due to the reduction from the near neighbor problem in all of $\R^d$ (under the $\ell_2$-norm) to solving the spherical case~\cite{andoni15}. Furthermore, using standard reductions the results for the $\ell_2$-norm translate to results for $\R^d$ with the $\ell_1$-norm as well.} has size $n = 2^{o(d)}$, and an unusually near neighbor lies at distance $r = \sqrt{2} / c$, the trade-offs we obtain for finding such a close vector to a random query vector are governed by the following relation.
\begin{theorem}[Time--space trade-offs for sparse data] \label{thm:main}
For $c > 1$ and $\rhoq, \rhos \geq 0$ satisfying
\begin{align}
(2c^2 - 1) \rhoq + 2 c^2 (c^2 - 1) \sqrt{\rhos (1 - \rhos)} \geq c^4, \label{eq:sparse}
\end{align}
there exists a graph-based $(c,r)$-ANN data structure for data sets of size $n = 2^{o(d)}$ on the sphere using space $n^{1 + \rhos + o(1)}$ and query time $n^{\rhoq + o(1)}$.
\end{theorem}
Minimizing the query complexity in this asymptotic analysis corresponds to setting $\rhos = \frac{1}{2}$, in which case $\rhoq = c^2 / (2c^2 - 1)$. Note that, unlike in various hash-based constructions, there is a natural limit to the maximum space complexity of graph-based approaches: we cannot store more than $n$ neighbors per vertex. Our analysis however suggests that when doing a greedy search in the graph, storing more than $\sqrt{n}$ neighbors per vertex does not further improve the query complexity, compared to starting over at a random new vertex.

To compare the above trade-offs with hash-based results, recall that in~\cite{andoni17} the authors obtained hash-based time--space trade-offs defined by the following inequality:
\begin{align}
c^2 \sqrt{\rhoq} + (c^2 - 1) \sqrt{\rhos} \geq \sqrt{2c^2 - 1}. \label{eq:hash}
\end{align}
Although in most cases the optimal\footnote{Within a certain probing model, the hash-based trade-offs from~\eqref{eq:hash} were proven to be optimal in~\cite{andoni17}.} hash-based trade-offs from~\eqref{eq:hash} are strictly superior to the graph-based trade-offs from Theorem~\eqref{thm:main}, we remark that in the regime of small approximation factors $c \approx 1$ and near-linear space $\rhos \approx 0$, i.e.\ when there is no unusually near neighbor and we wish to use only slightly more space than is required for storing the input list, both inequalities above translate to:
\begin{align}
\rhoq = 1 - 4 (c - 1) \sqrt{\rhos} \cdot (1 + o(1)).
\end{align}
Here $o(1)$ vanishes as $\rhos \to 0$ and $c \to 1$. So for truly random instances without any planted, unusually near neighbors, and when using a limited amount of memory, graph-based methods are asymptotically equally powerful for finding (approximate) nearest neighbors as optimal hash-based methods. In other words, in this regime graph-based approaches will remain competitive with hash-based methods even when the data sets become very large, and the dimensionality further increases. 

On the negative side, our analysis suggests that when there is an unusually near (planted) neighbor to the query, or when we are able to use much more space than the amount required to store the data set, the best known hash-based approaches are superior to the basic graph-based method analyzed in this paper. This could be caused either by our analysis not being tight, the considered algorithm not being as optimized, or due to graph-based approaches simply not being able to profit as much from such unusual circumstances. Note again that hash-based approaches are able to effectively use very large amounts of space, with a large number of fine-grained hash tables/partitions, whereas graph-based methods seem limited by using at most $n^2$ space. We therefore conjecture that the worse asymptotic performance for ``unusual'' problem instances is inherent to graph-based methods.

\subparagraph{Dense data sets.} For settings where the data set consists of $n = 2^{\Theta(d)}$ points drawn uniformly at random from the unit sphere, the asymptotic performance of near neighbor searching depends not only on the distance to the (planted) nearest neighbor, but also on the density of the data set, i.e.\ the exact relation between $d$ and $n$. Motivated by concrete data sets commonly used in practice, we concretely analyze the case $(\log n) / d \approx 1/5$ and show that the resulting trade-offs for the exact nearest neighbor problem without planted neighbors are significantly better than hyperplane LSH~\cite{charikar02}; comparable to or better than cross-polytope LSH~\cite{andoni15cp, falconn} and spherical cap LSH~\cite{andoni14}; but slightly worse than spherical LSF~\cite{becker16lsf, andoni17}. The superior asymptotic performance compared to \FALCONN{} for realistic dense data sets suggests that graph-based approaches may remain competitive with the most practical hash-based approaches, even when the dimensionality further increases.

\subparagraph{Future work.} Various open problems remain to obtain a better theoretical (and practical) understanding of different near neighbor techniques, in particular related to graph-based approaches. We state some remaining open problems below:
\begin{itemize}
\item Although our analysis for ``small steps'' (see Appendix~\ref{app:analysis}) is tight, we assumed that afterwards we either immediately find the planted nearest neighbor as one of the neighbors in the graph, or we fail and start over from a new random node. This seems rather pessimistic, and perhaps the complexities can be further improved with a tighter analysis, without changing the underlying algorithm or graph construction.
\item Other heuristic nearest neighbor graph approaches appear to perform even better in practice, and an open problem is to rigorously analyze those and see whether they translate to improved asymptotics as well.
\item A practical drawback of using nearest neighbor graphs is that updating the data structure (in particular: inserting new data points) can be rather costly, especially in comparison with hash-based constructions. To make the data structure both efficient and dynamic, one would ideally obtain better bounds on the insertion complexity as well.
\item In hash-based literature, the case of sparse data sets has arguably almost been ``solved'' with upper bounds matching lower bounds (within a certain model). Is it possible to find similar lower bounds for graph-based near neighbor searching?
\item Hash-based and graph-based approaches could be considered complementary solutions to the same problem, as worst-case problem instances for one approach are best-case instances for the other. Can both techniques be combined, so that perhaps even better worst-case guarantees can be obtained for high-dimensional data sets?
\end{itemize}

\subparagraph{Outline.} The remainder of this paper is organized as follows. In Section~\ref{sec:pre}, we introduce preliminary results and notation. Section~\ref{sec:random} describes problem instances considered in this paper. Section~\ref{sec:graph} analyzes the complexities of finding near neighbors with a greedy graph search, and Sections~\ref{sec:sparse} and \ref{sec:dense} consider asymptotics for sparse and dense data sets, respectively. The appendices contain proofs and discuss the application to lattice sieving.

%
%
%
%
%
%


\section{Preliminaries}
\label{sec:pre}

\subsection{Notation}

Let $\|\cdot\|$ denote the Euclidean norm, and let $\ip{\cdot}{\cdot}$ denote the standard dot product. We write $\cS^{d-1} = \{\vc{u} \in \R^d: \|\vc{u}\| = 1\}$ for the Euclidean unit sphere in $\R^d$. We write $X \sim \chi(\mathcal{X})$ to denote that the random variable $X$ is sampled from the probability distribution $\chi$ over the set $\mathcal{X}$, and we write $\U(\mathcal{X})$ for the uniform distribution over $\mathcal{X}$. Given a vector $\vc{x}$ (written in boldface), we further write $x_i = \ip{\vc x}{\vc{e}_i}$ for the $i$th coordinate of $\vc{x}$.

We denote directed graphs by $G = (V, A)$ where $V$ denotes the set of vertices, and $A \subseteq V \times V$ denotes the set of directed arcs. A directed graph is called symmetric if $(v_1, v_2) \in A$ iff $(v_2, v_1) \in A$. Symmetric directed graphs can also be viewed as undirected graphs $G = (V, E)$ where edges are unordered subsets of $V$ of size $2$. 

\subsection{Geometry on the sphere} 

Let $\C_{\vc{x}, \alpha} = \{\vc{u} \in \cS^{d-1}: \ip{\vc u}{\vc x} \geq \alpha\}$ denote the spherical cap centered at $\vc{x} \in \cS^{d-1}$ of ``height'' $\alpha \in (0,1)$, and let $C(\alpha)$ denote its volume relative to the entire unit sphere. Let $\W_{\vc{x}, \alpha, \vc{y}, \beta} = \C_{\vc{x}, \alpha} \cap \C_{\vc{y}, \beta}$ with $\vc{x}, \vc{y} \in \cS^{d-1}$ and $\alpha, \beta \in (0,1)$ denote the intersection of two spherical caps, and let its volume relative to the volume of the unit sphere be denoted $W(\alpha, \beta, \gamma)$, where $\gamma = \ip{\vc x}{\vc y}$ is the cosine of the angle between $\vc{x}$ and $\vc{y}$. We will also call the latter objects \textit{wedges}. The volumes of these objects correspond to probabilities on the sphere as follows:
\begin{align}
C(\alpha) &= \Pr_{\vc{X} \sim \U(\cS^{d-1})} (X_1 > \alpha), \\ 
W(\alpha, \beta, \gamma) &= \Pr_{\vc{X} \sim \U(\cS^{d-1})} (X_1 > \alpha, X_1 \gamma + X_2 \sqrt{1 - \gamma^2} > \beta). 
\end{align}
The asymptotic relative volumes of spherical caps and wedges can be estimated as follows, as previously shown in e.g.\ \cite{becker16lsf, andoni17}.
\begin{lemma}[Volume of a spherical cap]
Let $\alpha \in (0,1)$. Then:
\begin{align}
C(\alpha) &= d^{\Theta(1)} \cdot (1 - \alpha^2)^{d/2}.
\end{align}
\end{lemma}
\begin{lemma}[Volume of a wedge] \label{lem:wedge}
Let $\alpha, \beta, \gamma \in (0,1)$. Then:
\begin{align}
W(\alpha, \beta, \gamma) &= d^{\Theta(1)} \cdot \begin{cases} \left(\dfrac{1 - \alpha^2 - \beta^2 - \gamma^2 + 2 \alpha \beta \gamma}{1 - \gamma^2}\right)^{d/2} \quad & \text{if } 0 < \gamma \leq \min\{\frac{\alpha}{\beta}, \frac{\beta}{\alpha}\}; \\ (1 - \alpha^2)^{d/2} & \text{if } \frac{\beta}{\alpha} \leq \gamma < 1; \\ (1 - \beta^2)^{d/2} & \text{if } \frac{\alpha}{\beta} \leq \gamma < 1. \end{cases}
\end{align}
\end{lemma}
\begin{figure}[!t]
\begin{center}
\includegraphics[width=8cm]{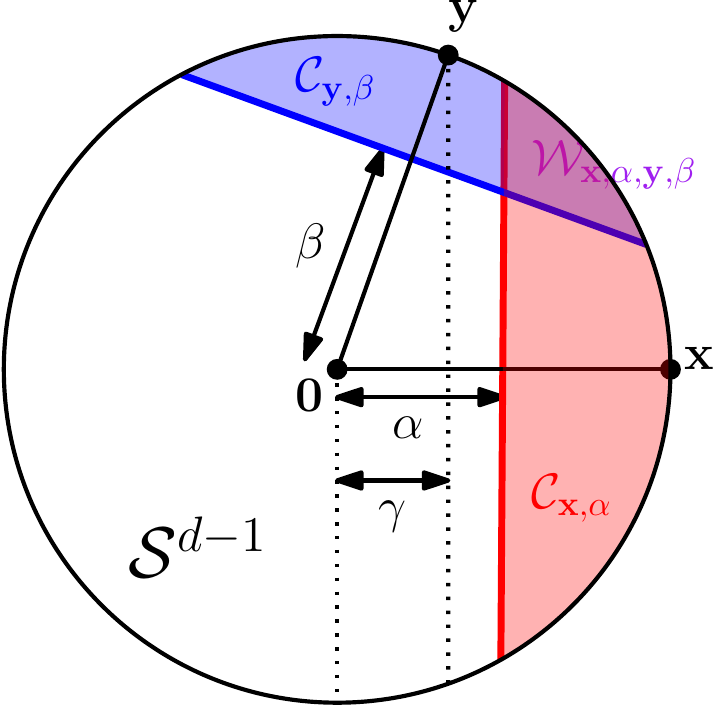}
\end{center}
\vspace{-10pt}
\caption{Geometry on the sphere. The relative volume of the wedge, $\mu(\W_{\vc{x}, \alpha, \vc{y}, \beta}) / \mu(\cS^{d-1})$, is denoted $W(\alpha, \beta, \gamma)$, with $\gamma$ (as in the sketch) denoting the cosine of the angle between $\vc{x}$ and $\vc{y}$.\label{fig:sphere}}
\end{figure}
For the wedge, the volume can alternatively be described (up to polynomial factors in $d$) as the volume of a spherical cap with height $\delta = \sqrt{\frac{\alpha^2 + \beta^2 - 2 \alpha \beta \gamma}{1 - \gamma^2}}$. In Figure~\ref{fig:sphere}, this $\delta$ corresponds to the smallest distance from points in $\W_{\vc{x}, \alpha, \vc{y}, \beta}$ to the origin, i.e.\ the distance from the intersection of the blue and red lines in this projection of the sphere to the origin. In case $\gamma \geq \frac{\alpha}{\beta}$ with $\alpha \leq \beta$, this distance from the origin becomes $\delta = \beta$ and therefore $W(\alpha, \beta, \gamma) = d^{\Theta(1)} \cdot C(\beta)$. In that case, one essentially has $\C_{\vc{x}, \alpha} \cap \C_{\vc{y}, \beta} \approx \C_{\vc{y}, \beta}$. 

We will be using various properties of these wedges, and we state some of them below.

\begin{lemma}[Wedge properties] For $\alpha, \beta, \gamma \in (0,1)$ with $\gamma < \min\{\frac{\alpha}{\beta}, \frac{\beta}{\alpha}\}$:
\begin{enumerate}
\item For $d$ sufficiently large, $W(\alpha, \beta, \gamma)$ is decreasing with $\alpha, \beta$ and increasing with $\gamma$;
\item For $d$ sufficiently large, $W(\alpha, \beta, \beta)$ is decreasing with $\beta$.
\end{enumerate}
\end{lemma}

\begin{proof}
For 1.\ the result follows by taking derivatives w.r.t.\ $\alpha, \beta, \gamma$ of the ``asymptotic part'' of $W$ (i.e.\ ignoring the $d^{\Theta(1)}$ term), and observing that these derivatives are always negative, negative, and positive respectively. For 2.\ we take the derivative w.r.t.\ $\beta$ of $W(\alpha, \beta, \beta)$ and observe that this derivative is always negative.
\end{proof}

The following lemma further describes that if we add a small amount of ``slack'' to one of the variables, then the volume of the resulting wedge will still be very similar to the volume of the wedge with the original parameters -- if the slack is small, then we only lose at most a polynomial factor in the volume.

\begin{lemma}[Polynomial slack] \label{lem:slack} For fixed $\alpha, \beta, \gamma \in (0,1)$ with $\gamma < \min\{\frac{\alpha}{\beta}, \frac{\beta}{\alpha}\}$ we have 
\begin{align}
W(\alpha, \beta \pm \tfrac{1}{d}, \gamma) = d^{\mp\Theta(1)} \cdot W(\alpha, \beta, \gamma), \qquad \text{and } W(\alpha, \beta, \gamma \pm \tfrac{1}{d}) = d^{\pm \Theta(1)} W(\alpha, \beta, \gamma).
\end{align}
\end{lemma}

\begin{proof}
This follows from writing out the left hand sides, pulling out the factors $W(\alpha, \beta, \gamma)$, and noting that the remaining factor $(1 \pm \eps_d)^d$ for some expression $\eps_d$ is at most polynomial in $d$, due to the conditions on $\gamma$ and $\alpha, \beta, \gamma$ being constant in $d$.
\end{proof}


\section{Random instances}
\label{sec:random}

For the graph-based approach in this paper, where points are connected to their nearest neighbor and we perform a walk on this graph starting from a random node, it is impossible to solve worst-case instances of (approximate) nearest neighbor searching. 

\begin{prop}[Worst-case data sets]
For near neighbor graph approaches, where (i) each vertex is only connected to a number of its nearest neighbors, and (ii) queries are answered by performing a greedy search on this graph to obtain better estimates, it is impossible to solve worst-case (approximate) near neighbor instances in sublinear time.
\end{prop}

\begin{proof}
As a potential worst-case problem instance to such a strategy, consider a query $\vc{q} \approx \vc{e}_1$, a planted nearest neighbor $\vc{p}^* \approx \vc{e}_1$ close to the query, and let all other points $\vc{p} \in \cD \setminus \{\vc{p}^*\}$ satisfy $\vc{p} \approx -\vc{e}_1$. Then the preprocessed near neighbor graph will consist of a large connected component containing the points $\cD \setminus \{\vc{p}^*\}$, and an isolated vertex $\vc{p}^*$, which may have outgoing edges, but has no mutual friends. With probability $1 - 1/n$, starting at a random vertex and performing a walk on this (directed) graph will therefore not yield the true nearest neighbor $\vc{p}^*$, while all other vertices are only approximate solutions with very high approximation factors; by essentially setting $\vc{p}^* = \vc{q}$, we can guarantee that even no reasonable approximate solution will be found.
\end{proof}

Although it may be possible to tweak the algorithm and/or the underlying graph so that even such worst-case instances can still be solved efficiently, we will therefore focus on average-case, random instances defined below.

Throughout, we will assume that points $\vc{p} \in \cD$ are independently and uniformly distributed on the unit sphere, except for (potentially) a planted nearest neighbor $\vc{p}^* \in \cD$ which lies very close to the query vector $\vc{q}$. Alternatively, this can be interpreted as taking a uniformly random $\vc{p}^* \sim \U(\cD)$, and choosing the query vector as $\vc{q} = \vc{p}^* + \vc{e}$ for a short error vector $\vc{e}$. Given such a problem instance, we wish to recover $\vc{p}^*$ with non-negligible probability. 

Note that the near neighbor problem on the Euclidean unit sphere, as considered here, is of special interest as such a solution allows one to the Euclidean near neighbor problem in all of $\R^d$ using techniques of Andoni--Razenshteyn~\cite{andoni15}. Furthermore, it is well-known that using standard reductions, the results for the $\ell_2$-norm translate to results for $\R^d$ with the $\ell_1$-norm as well. Efficient algorithms for the unit sphere therefore translate to solutions for many other problems as well.

For $\vc{q} \sim \U(\cD)$, and data sets following a uniformly random distribution, with overwhelming probability the (non-planted) second nearest neighbor $\vc{p} \in \cD \setminus \{\vc{p}^*\}$ to $\vc{q}$ has inner product $\ip{\vc{p}}{\vc q} = \mu (1 + o(1))$ and lies at distance $\|\vc{q} - \vc{p}\| = \sqrt{2 (1 - \mu)} (1 + o(1))$ from $\vc{q}$, with $\mu$ satisfying:
\begin{align}
\mu = \sqrt{1 - n^{-2/d}}. \qquad \qquad (\text{Alternatively: $n \approx 1 / C(\mu)$.})
\end{align}
The natural scenario to consider for random ANN instances is then to let $c \cdot r = \mu$, so that the single planted nearest neighbor lies at distance $r = \mu / c$ from the target, and so that this planted near neighbor lies a factor $c$ closer to $\vc{q}$ than all other points in the data set. These problem instances were also studied in for example~\cite{andoni15cp, andoni17}.

\subparagraph{Sparse data sets.} We will refer to data sets of size $n = 2^{o(d)}$ (in other words: $d = \omega(\log n)$) as \textit{sparse} data sets. For these instances, the expected (non-planted) nearest neighbor distance is $\mu = \sqrt{2} + o(1)$, and the planted nearest neighbor which we wish to recover therefore lies at distance $r = \sqrt{2}/c \, (1 + o(1))$ from the target. Note that the case $n = 2^{o(d / \log d)}$ can always be reduced to $n = 2^{\Theta(d / \log d)}$ through a random projection onto a lower-dimensional space, approximately maintaining all pairwise distances between points in the data set~\cite{johnson84}.

\subparagraph{Dense data sets.} We will refer to data sets of size $n = 2^{\Theta(d)}$ ($d = \Theta(\log n)$) as \textit{dense} data sets. In this case, with overwhelming probability the nearest neighbor to a randomly chosen query point $\vc{q}$ on the sphere will lie at distance $c \cdot r = \mu < \sqrt{2}$ from $\vc{q}$, and for $(c,r)$-ANN the planted nearest neighbor would therefore lie at distance $r = \mu / c$. For the limiting case of $c \to 1$, we wish to recover a vector at distance $\mu$. Note that the ``curse of dimensionality''~\cite{indyk98} does not necessarily apply to average-case dense data sets -- finding exact nearest neighbors for random dense data sets can often be done in sub-linear time~\cite{laarhoven15nns, becker16lsf}.
 

\section{Graph-based near neighbor searching}
\label{sec:graph}

\subsection{Algorithm description}

The nearest neighbor graph and corresponding near neighbor search algorithm we will analyze are very similar to a greedy search in the $k$-nearest neighbor graph, i.e.\ very similar to the approach of \KGraph{}. Recall that the directed $k$-nearest neighbor graph $G = (V, A)$ has vertex set $V = \cD$, and an arc runs from $\vc{p}$ to $\vc{p}'$ iff $\vc{p}'$ belongs to the $k$ closest points to $\vc{p}$ in $\cD$. Given a query $\vc{q}$, searching for a nearest vector in this graph is commonly done as outlined in Algorithm~\ref{alg:query}, with $\eps = 0$, and with $\B(\vc{p})$ denoting the set of neighbors to $\vc{p}$ in this graph. Note that for large $k = n^{\Theta(1)}$, this graph is almost symmetric. The condition of belonging to the $k$ nearest neighbors is however somewhat impractical to work with analytically, and so we will use the following slightly different graph (and search algorithm) instead.

\begin{definition}[The $\alpha$-near neighbor graph]
Let $\alpha \in (0,1)$, and let $\cD \subset \cS^{d-1}$. We define the $\alpha$-near neighbor graph as the undirected graph $G = (V, E)$ with vertex set $V = \cD$, and with an edge between $\vc{p}, \vc{p}' \in E$ if and only if $\ip{\vc p}{\vc p'} \geq \alpha$. 
\end{definition}

The parameter $\alpha$ roughly corresponds to (a function of) $k$: large $\alpha$ correspond to small $k$, and small $\alpha$ to large $k$. Asymptotically, the approximate relation between $k$ and $\alpha$ can be stated through the following simple relation $k \approx n \cdot C(\alpha)$. The main difference is that rather than fixing $k$ and varying the required distance between points for an edge, we fix a bound on the distance between two connected points in the graph, and therefore we will have slight variations in the number of neighbors $k$ from vertex to vertex. Notice the similarity with spherical cap LSH~\cite{andoni14} and in particular spherical LSF~\cite{becker16lsf, andoni17}: we essentially use $n$ random spherical filters centered around our data points, and we add vectors to filter buckets the same way as in spherical LSF: if the inner product with the filter vector $\vc{p}$ is sufficiently large, we add the point to this bucket $\B(\vc{p})$. However, in the spirit of graph-based approaches we search for a path on the nearest neighbor graph that ends at the nearest neighbor as in Algorithm~\ref{alg:query}, rather than checking a number of filters close to the query point to see if the nearest neighbor is contained in any of those.

To process a query $\vc{q}$, we first sample a uniformly random point $\vc{p} \sim \U(\cD)$. Then we go through its neighbors $\B(\vc{p})$ to see if any of these vectors $\vc{p}' \in \B(\vc{p})$ are closer to $\vc{q}$ than $\vc{p}$. If so, we use this as our new $\vc{p} \gets \vc{p}'$, and we again see if any of its neighbors are closer to $\vc{q}$. We repeat this procedure until no more vectors in $\B(\vc{p})$ are closer to $\vc{q}$ than $\vc{p}$ itself, in which case $\vc{p}$ is our estimate for the real nearest neighbor $\vc{p}^*$ to $\vc{q}$. This may ($\vc{p} = \vc{p}^*$) or may not ($\vc{p} \neq \vc{p}^*$) actually be the true nearest neighbor to $\vc{q}$, and to obtain a higher success rate we repeat the process several times, each time starting from a random point $\vc{p} \sim \U(\cD)$.

While the focus is on minimizing the query time, Algorithms~\ref{alg:insert}--\ref{alg:delete} also demonstrate how to do updates to this data structure, when vectors are inserted in $\cD$ or removed from $\cD$. These parts of the algorithm may well be improved upon, and an important open question is to make updates (in particular insertions) as efficient as partition-based methods, such as in~\cite{andoni17}.

\begin{figure}[p]
\begin{minipage}[t]{7.1cm}
\vspace{0pt}
\begin{algorithm}[H]
\caption{\ $\alpha$-NN graph: \textsc{Query}($\vc{q}$) \label{alg:query}}
\begin{algorithmic}[1]
\State $\vc{p} \sim \U(\cD)$ \Comment{random starting point}
\While{$\ip{\vc{p}}{\vc{q}} < \gamma^*$} \Comment{$\gamma^* = \ip{\vc{p}^*}{\vc{q}}$}
	\State \textsf{progress} $\gets$ \textbf{false}
	\For{\textbf{each} $\vc{p}' \in \B(\vc{p})$} \label{lin:for}
		\If{$\ip{\vc{p}'}{\vc q} \geq \ip{\vc p}{\vc q} + \frac{1}{d}$}
			\State $\vc{p} \gets \vc{p}'$ \Comment{nearer neighbor}
			\State \textsf{progress} $\gets$ \textbf{true}
		\EndIf
	\EndFor
	\If{\textbf{not} \textsf{progress}}
		\State $\vc{p} \sim \U(\cD)$ \label{lin:restart} \Comment{start over}
	\EndIf
\EndWhile
\State \Return $\vc{p}$
\end{algorithmic}
\end{algorithm}
\end{minipage}
\,
\begin{minipage}[t]{6.5cm}
\vspace{0pt}  
\begin{algorithm}[H]
\caption{\ $\alpha$-NN graph: \textsc{Insert}($\vc{p}$) \label{alg:insert}}
\begin{algorithmic}[1]
\State $\B(\vc{p}) \gets \emptyset$ \Comment{new bucket}
\For{\textbf{each} $\vc{p}' \in \cD$}
	\If{$\ip{\vc p}{\vc p'} \geq \alpha$}
		\State $\B(\vc{p}) \gets \B(\vc{p}) \cup \{\vc{p}'\}$
		\State $\B(\vc{p}') \gets \B(\vc{p}') \cup \{\vc{p}\}$
	\EndIf
\EndFor
\end{algorithmic}
\end{algorithm}
\vspace{-18pt}
\begin{algorithm}[H]
\caption{\ $\alpha$-NN graph: \textsc{Delete}($\vc{p}$) \label{alg:delete}}
\begin{algorithmic}[1]
\For{\textbf{each} $\vc{p}' \in \B(\vc{p})$}
	\State $\B(\vc{p}') \gets \B(\vc{p}') \setminus \{\vc{p}\}$
\EndFor
\State $\B(\vc{p}) \gets \emptyset$ \Comment{delete bucket}
\end{algorithmic}
\end{algorithm} 
\end{minipage} 
\caption{Algorithms for querying the $\alpha$-near neighbor graph with a query point $\vc{q}$, and for inserting/deleting points from this data structure. Here $\gamma^*$ is the (expected) inner product between $\vc{q}$ and its true nearest neighbor. Initializing the data structure is done by choosing $\alpha \in (0,1)$ and for instance calling \textsc{Insert}($\vc{p}$) for all $\vc{p} \in \cD$ to construct the graph adjacency buckets $\B(\vc p)$.}
\end{figure}

\begin{figure}[p]
\centering
\includegraphics[width=13.5cm]{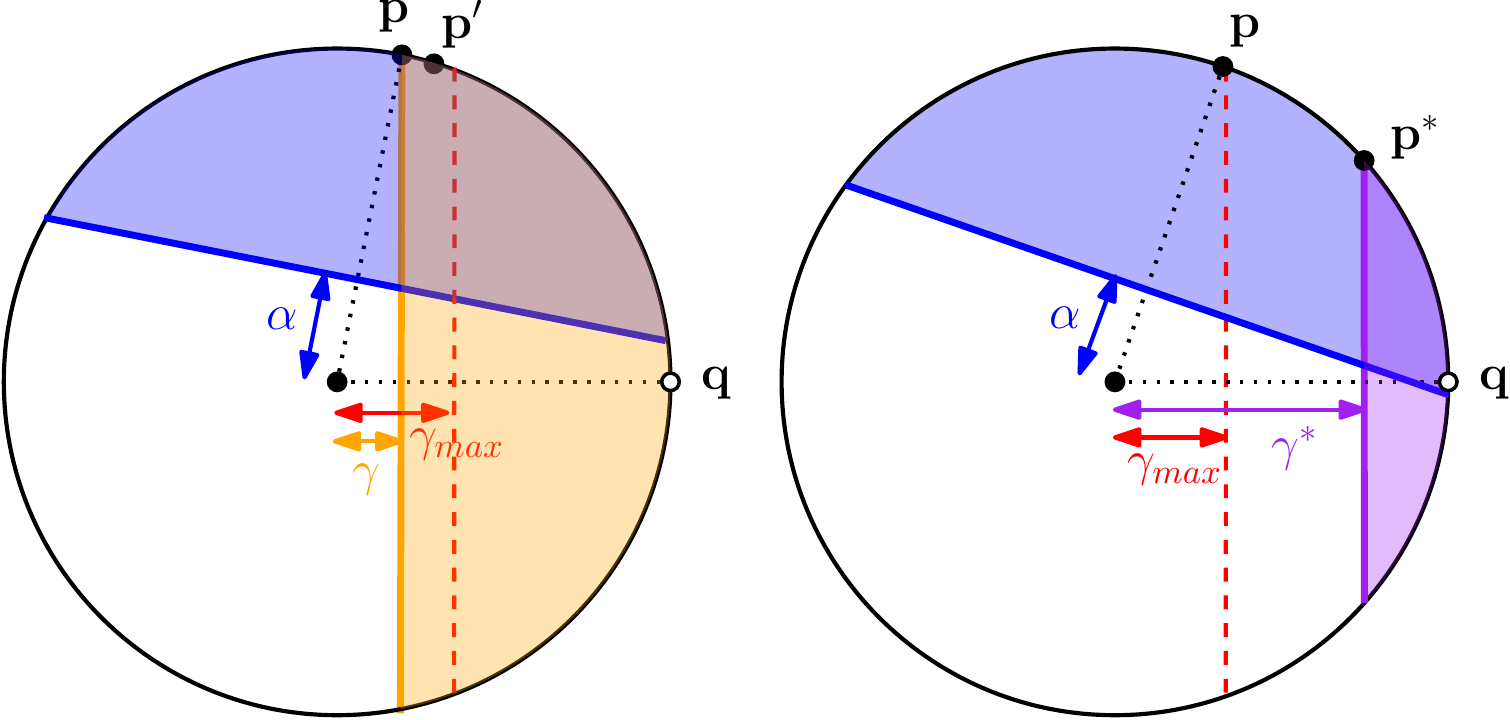}
\caption{A sketch of the analyses for small steps (left) and the giant leap (right). The point $\vc{p} \in \cD$ denotes the current near neighbor estimate for the query $\vc{q}$, and we want to make progress either by finding a \textit{slightly nearer} neighbor $\vc{p}' \in \B(\vc{p})$ when $\vc{p}$ is \textit{far away} from $\vc{q}$ (left), or finding the \textit{exact} nearest neighbor $\vc{p}^* \in \B(\vc{p})$ to the query $\vc{q}$ when $\vc{p}$ is already quite close to $\vc{q}$ (right). The threshold separating these two cases is $\gamma_{max}$. \protect\\
\textbf{$\bullet$ Small steps (left).} If the current near neighbor $\vc{p}$ has inner product $\ip{\vc p}{\vc q} = \gamma \ll \gamma_{max}$ with $\vc{q}$, then we expect that several points $\vc{p}' \in \cD \setminus \{\vc{p}^*\}$ still exist which lie (slightly) closer to $\vc{q}$ than $\vc{p}$, and we hope at least one of them is an $\alpha$-near neighbor to $\vc{p}$ as well. Since a nearer neighbor in $\B(\vc{p})$ to $\vc{q}$ by definition lies in $\W_{\vc{p}, \alpha, \vc{q}, \ip{\vc p}{\vc q}}$ (the intersection of the solid spherical caps), and the data set is assumed to be uniformly random on the sphere, the probability of finding at least one such nearer neighbor is proportional to $n \cdot W(\alpha, \gamma, \gamma)$.\protect\\
\textbf{$\bullet$ Giant leap (right).} Once we find a near neighbor $\vc{p}$ to $\vc{q}$ with inner product $\ip{\vc{p}}{\vc{q}} \approx \gamma_{max}$, we would like to show that in the next step, we will find $\vc{p}^* \in \B(\vc{p})$ with a certain (small) probability. Assuming $\vc{p}^*$ is uniformly distributed on $\C_{\vc{q}, \gamma^*}$, this corresponds to the probability that $\vc{p}^* \in \W_{\vc{p}, \gamma_{max}, \vc{q}, \gamma^*}$ (the intersection of the solid spherical caps), conditioned on the event $\vc{p}^* \in \C_{\vc{q}, \gamma^*}$ (the rightmost spherical cap). This probability is therefore proportional to $W(\alpha, \gamma^*, \gamma_{max}) / C(\gamma^*)$. \label{fig:sketch}}
\end{figure}

\subsection{High-level proof description}

To obtain provable asymptotic complexities for this approach for the (approximate) nearest neighbor problem, we will prove the following results in Appendix~\ref{app:analysis} for some value $\gamma_{max}$:
\begin{description}
\item[$\bullet$ Small steps:] If $\vc{p} \in \cD$ satisfies $\ip{\vc{p}}{\vc q} \ll \gamma_{max}$, then with non-negligible probability $d^{-\Theta(1)}$ we will find a slightly nearer neighbor $\vc{p}' \in \B(\vc{p})$ to $\vc{q}$.
\item[$\bullet$ Giant leap:] If $\vc{p} \in \cD$ satisfies $\ip{\vc{p}}{\vc q} \approx \gamma_{max}$, the probability of immediately finding the exact nearest neighbor $\vc{p}^*$ in the bucket $\B(\vc{p})$ is larger than some given bound.
\item[$\bullet$ Randomization:] Since ``giant leaps'' may often fail, we argue that starting over at random nodes a sufficiently large number of times leads to a constant success probability.
\item[$\bullet$ Encountered vertices:] We prove that the number of vertices in each bucket, and on each walk through the graph, can be bounded by small multiples of their expected values.
\item[$\bullet$ Query complexity:] Using these results, we derive bounds on the time complexity for answering queries, with and without starting over at random nodes in the graph.
\item[$\bullet$ Space complexity:] Similarly, since the number of edges in the graph can be bounded appropriately, we obtain tight bounds on the required space complexity.
\item[$\bullet$ Update complexities:] Finally, we analyze what are the (naive) costs for updating the data structure (inserting/deleting points).
\end{description}
Together, these results lead to the following main result, stating exactly what are the costs for finding near neighbors. Unless stated otherwise, the ``time'' complexity corresponds to the \textit{query} time complexity.

\begin{theorem}[Near neighbor costs]
Using the $\alpha$-near neighbor graph with $\alpha \in (0,1)$ and $n C(\alpha) \gg 1$, and using Algorithms~\ref{alg:query}--\ref{alg:delete}, we can solve the planted near neighbor problem with the following costs, with $\gamma_{max}$ as in~\eqref{eq:gammax}:
\begin{align}
\texttt{Time} &= d^{O(1)} n C(\alpha) C(\gamma^*) / W(\alpha, \gamma^*, \gamma_{max}), \\
\texttt{Space} &= O(n^2 C(\alpha) \log n), \\
\texttt{Insert} &= O(d n), \\
\texttt{Delete} &= O(n C(\alpha) \log (n C(\alpha))).
\end{align}
\end{theorem}


\section{Asymptotics for sparse data sets}
\label{sec:sparse}

Due to space limitations, the derivation of the asymptotics on the various costs of graph-based near neighbor searching for large $d$ have been moved to Appendix~\ref{app:sparse}. Its derivation mainly consists of carefully writing out the costs stated in the previous theorem, and doing series expansions for $\mu = \sqrt{1 - n^{-2/d}} = o(1)$.

\begin{theorem}[Complexities for sparse data]
Let $n = 1 / C(\mu) = 2^{o(d)}$ and let $\alpha = \kappa \cdot \mu$. Let $\gamma^* = 1 - 1/c^2$ denote the inner product between the query and the (planted) nearest neighbor. Using the $\alpha$-NNG with $\alpha = \kappa \cdot \mu$ with $\kappa \in (\sqrt{(\gamma^*)^2/ (1 + (\gamma^*)^2)}, 1)$, with high probability we can solve the sparse near neighbor problem in several iterations with the following complexities:
\begin{align}
\texttt{Time} &= n^{\left[\frac{1 - 2 \gamma^* \sqrt{\kappa^2 \left(1 - \kappa^2\right)}}{1 - (\gamma^*)^2} + o(1)\right]}, \\
\texttt{Space} &= n^{2 - \kappa^2 + o(1)}, \\
\texttt{Insert} &= n^{1 + o(1)}, \\
\texttt{Delete} &= n^{1 - \kappa^2 + o(1)}.
\end{align}
Denoting the query time complexity by $\texttt{Time} = n^{\rhoq + o(1)}$ and the space complexity by $\texttt{Space} = n^{1 + \rhos + o(1)}$, the trade-off between these costs can be expressed using the following inequality:
\begin{align}
(2c^2 - 1) \rhoq + 2 c^2 (c^2 - 1) \sqrt{\rhos (1 - \rhos)} \geq c^4.
\end{align}
\end{theorem}
As described in the introduction, for $c \approx 1$ and $\rhos \to 0$, the above condition on $\rhoq$ scales as $\rhoq \geq 1 + 4 (c - 1)^2 + O((c - 1)^4)$, which is equivalent to the scaling near $c = 1$ for the optimal partition-based near neighbor method of~\cite{andoni17}. For larger $c$ and when using more space, this trade-off is not better than the best hash-based methods -- by substituting $\sqrt{\rhos (1 - \rhos)} \leq \frac{1}{2}$, we obtain the necessary (but not sufficient) condition $\rhoq \geq c^2 / (2c^2 - 1)$, which shows that the query complexity never reduces beyond $\sqrt{n}$, even for large $c$. This inability to ``profit'' from large approximation factors may be inherent to graph-based approaches, or at least to the greedy graph-based approach considered in this paper.

For the ``balanced'' trade-off of $\rhoq = \rhos = \rho$, the condition on the exponents translates to:
\begin{align}
\rho \geq \frac{c^4}{2c^4 - 2c^2 + 1} \, .
\end{align}
For small $c \approx 1$, this leads to the asymptotic scaling $\rho = 1 - 4 (c - 1)^2 + O((c - 1)^3)$, which is slightly worse than the optimal hash-based trade-offs of~\cite{andoni14, andoni17}: $\rho = 1/(2c^2 - 1) = 1 - 4 (c - 1) + 14 (c - 1)^2 + O((c - 1)^3)$. Also note again the asymptotics of $\rho \to \frac{1}{2}$ for large $c$ for graph-based methods,


\section{Asymptotics for dense data sets}
\label{sec:dense}

For data sets of size $n = 2^{\Theta(d)}$, the asymptotics from the previous section do not apply; these assumed that $\mu = \sqrt{1 - n^{-2/d}} = o(1)$. In some applications the relation $d = \Theta(\log n)$ is more accurate, and arguably even if e.g.\ $d = \Theta(\log n \log \log n)$ grows faster than $\log n$, asymptotics for the sparse regime may be rather optimistic; $\log \log n$ terms are then considered ``large'', even though for realistic parameters $\log \log n$ may just as well be considered constant.

In Appendix~\ref{app:sieving}, we performed a case study for an application where $\mu = \sqrt{1 - n^{-2/d}} = \frac{1}{2}$, so that $n \approx 2^{d/5}$. This example is motivated by algorithms for solving hard lattice problems and cryptanalyzing lattice-based cryptosystems, which have previously been improved using other (hash-based) near neighbor methods~\cite{laarhoven15crypto, becker16cp, becker16lsf}. Comparable results may hold for other applications as well; the GloVe data set~\cite{pennington14} contains $n = 1.2M \approx 2^{20}$ vectors in $d = 100$ dimensions, which corresponds to $n \approx 2^{d/5}$, while the SIFT features data set~\cite{amsaleg10} has $n = 1M$ vectors in $d = 128$ dimensions, corresponding to $n \approx 2^{0.35d}$. The conclusions regarding the comparison of graph-based near neighbor searching techniques with hash-based methods (in Appendix~\ref{app:sieving}) may therefore apply to other data sets and applications as well.

Note that since the precise trade-offs depend on many parameters (see Theorem~\ref{thm:main}), we do not state a main result here, and refer the reader to Appendix~\ref{app:sieving} for one specific case analysis for lattice sieving. For arbitrary densities and target angles, this is a matter of writing out the expressions from Theorem~\ref{thm:main} and looking at the asymptotic scaling of the various costs.


\section*{Acknowledgments}

The author thanks Marek Eli\'{a}\v{s}, Jesper Nederlof, and Jorn van der Pol for enlightening discussions and comments about the proofs in the appendix. This work was supported by ERC consolidator grant 617951.


\appendix


\section{Proof of Theorem~\ref{thm:main}}
\label{app:analysis}

Here we will give a concrete theoretical analysis of the costs associated to the $\alpha$-near neighbor graph, where for each pair of vertices (points) in the uniformly random data set $\cD$, we connect them by an edge if their joint inner product is at least $\alpha > 0$.

\subsection{Small steps towards the goal}

First, we want to show that with high probability, when starting from a random vertex, we can easily find nearer neighbors until we reach a point $\vc{p}' \in \cD$ with $\ip{\vc{p}'}{\vc q} \geq \gamma_{max}$. This threshold $\gamma_{max}$ is determined by the density of the data set ($\mu$) and the density of the graph ($\alpha$), and when $\ip{\vc{p}'}{\vc q} > \gamma_{max}$ the probability of finding a nearer neighbor to $\vc{q}$ in the bucket of neighbors of $\vc{p}'$ becomes asymptotically negligible, i.e.\ exponentially small in $d$. In a sense a phase transition takes place at inner product $\gamma_{max}$, since from there on finding nearer neighbors becomes unlikely.

To formalize this, using the relation between probabilities on random data sets and volumes of (intersections of) spherical caps as defined in the preliminaries, we want to show that the volume of the wedge formed by the intersection of the blue and orange spherical caps in Figure~\ref{fig:sketch} scales as $1 / n$, so that, since we have $n$ i.i.d.\ uniformly random data points spread out over the sphere, one of them is likely to be contained in this intersection. Note that if a point lies in this intersection, it is both closer to the query point $\vc{q}$ (the orange cap) and it is a neighbor to $\vc{p}$ in the graph (the blue cap).

\begin{lemma}[High-probability reachable $\gamma$] \label{lem:gammax}
Let $\gamma_{max} = \gamma_{max}(\mu, \alpha)$ be defined as:
\begin{align}
\gamma_{max} = (1 + o(1)) \sqrt{\frac{\mu^2 - \alpha^2}{1 - 2\alpha + \mu^2}} \, . \label{eq:gammax}
\end{align}
Then $W(\alpha, \gamma_{max}, \gamma_{max}) \geq d^{\Theta(1)} C(\mu)$.
\end{lemma}

\begin{proof}
This is a matter of writing out the ``core part'' of $W(\alpha, \gamma_{max}, \gamma_{max})$, observing that many terms cancel due to the two equal arguments and the relation between $\gamma_{max}$ and $\alpha$, until finally we are left with $C(\mu)$, up to polynomial factors.
\end{proof}

For the next lemma, we make a slight modification to the query algorithm as well; rather than greedily searching for \textit{any} vector $\vc{p}' \in \B(\vc{p})$ which is closer to $\vc{q}$ than $\vc{p}$, we only look for vectors which make ``substantial progress'', i.e.\ for which $\ip{\vc{p}'}{\vc{q}} > \ip{\vc{p}}{\vc{q}} + \eps$ for some well-chosen $\eps$. With $\eps = 0$, the number of iterations to ultimately reach the correct nearest neighbor could theoretically be as large as $n$, while choosing $\eps$ too large means we will not easily make progress. By choosing $\eps = \frac{1}{d}$ we guarantee that only a linear number of steps is required to either reach our goal or terminate without success, and this choice of $\eps$ does not come at a significant asymptotic cost in the performance (essentially due to Lemma~\ref{lem:slack}).

\begin{lemma}[Collisions for nearer neighbors] \label{lem:small}
Let $\vc{p}, \vc{q}$ be fixed with $\ip{\vc{p}}{\vc{q}} + \frac{1}{d} < \gamma_{max}$ with $\gamma_{max}$ as defined in~\eqref{eq:gammax}. Then:
\begin{align}
p_1 = 
\Pr_{\vc{p}' \sim \U(\cS^{d-1})} \Big(\ip{\vc{p}'}{\vc{p}} \geq \alpha \text{ and } \ip{\vc{p}'}{\vc q} \geq \ip{\vc{p}}{\vc q} + \tfrac{1}{d}\Big) \geq \frac{\log (100 d)}{n}\, .
\end{align}
It follows that $p_n = \Pr(\exists \vc{p}' \in \B(\vc{p}): \ip{\vc{p}'}{\vc q} \geq \ip{\vc{p}}{\vc q} + \frac{1}{d}) \geq 1 - \frac{0.01}{d}$ for large $d$.
\end{lemma}

\begin{proof}
Let $\gamma = \ip{\vc{p}}{\vc q}$. First, observe that the stated probability $p_1$ equals the mass of a spherical wedge with parameters $\alpha$, $\gamma + \frac{1}{d}$ and $\gamma$. Since $\gamma + \tfrac{1}{d} \leq \gamma_{max}$, we can apply Lemma~\ref{lem:slack}:
\begin{align}
p_1 = W(\alpha, \gamma + \tfrac{1}{d}, \gamma) &\stackrel{(a)}{\geq} d^{-\Theta(1)} W(\alpha, \gamma + \tfrac{1}{d}, \gamma + \tfrac{1}{d}) \geq d^{-\Theta(1)} W(\alpha, \gamma_{max}, \gamma_{max}).
\end{align}
Using the previous lemma, we then obtain:
\begin{align}
p_1 \geq d^{-\Theta(1)} W(\alpha, \gamma_{max}, \gamma_{max}) \stackrel{(b)}{\geq} d^{-\Theta(1)} C(\mu) = \frac{d^{-\Theta(1)}}{n}. 
\end{align}
Now, note that the leading polynomial factor in the lower bound for $p_1$ comes from the hidden polynomial term in the volumes of spherical caps and intersections of spherical caps (where the exponent in $d^{\Theta(1)}$ is fixed and can be explicitly determined, using techniques as in~\cite[Lemma 4.1]{micciancio10b}), and the loss factor in Lemma~\ref{lem:slack} (see (a), where the exponent of the polynomial term can be tuned by slightly changing the slack to e.g.\ $1/d^{\beta}$ instead, for larger $\beta$\footnote{Note that changing the slack parameter from $1/d$ to $1/d^{\beta}$ for some constant $\beta > 1$ will ultimately change the required bound for the next subsection to be $p_1 \geq \log(100 d^{\beta}) / n$; a change of $\poly(d)$ here to the slack parameter only leads to a $\polylog(d)$ change there. It is therefore certainly possible to choose $\beta$ sufficiently large to make both this proof and the next work with low overhead.}. Furthermore, if necessary, we can introduce some ``slack'' in $\gamma_{max}$ as well (i.e. add a multiplicative factor $1 + \tfrac{1}{d}$ to $\gamma_{max}$), so that inequality (b) introduces polynomial factors in the ``right'' direction. By choosing these parameters appropriately, we can therefore tune the (log-)polynomial term in the right hand side to obtain e.g.\ $\log(100d)/n$. 

Finally, observe that if $p_1 \geq \log(100 d) / n$, then it follows that $p_n = 1 - (1 - p_1)^n \geq 1 - (1 - \frac{\log 100d}{n})^n \geq 1 - 0.01/d$ by independence of all (other) points in the data set.
\end{proof}

\begin{lemma}[Convergence towards $\gamma_{max}$]
Starting from an arbitrary $\vc{p} \in \cD$, with probability at least $0.98$ we will find a neighbor $\vc{p}'$ with $\ip{\vc p'}{\vc q} + \frac{1}{d} \geq \gamma_{max}$ within $2 d$ steps in the graph as in Algorithm~\ref{alg:query} (without starting over).
\end{lemma}

\begin{proof}
We start at a node $\vc{p} \in \cD$ with inner product $\ip{\vc p}{\vc q} \geq -1$, and we aim for an inner product $\ip{\vc p'}{\vc q} = \gamma_{max} \leq 1$, making steps of size $\frac{1}{d}$ each time. So clearly within $2 d$ iterations we will either reach our goal or fail. Since each step fails with probability at most $\frac{0.01}{d}$ by the previous lemma, by the union bound the probability that at least one step fails is at most $0.02$, which implies that all steps succeed with probability at least $0.98$.
\end{proof}

\subsection{A giant leap towards the nearest neighbor}

When, after making a number of small steps, we finally arrive at a vector $\vc{p} \in \cD$ with inner product approximately $\gamma_{max}$ with the query, we want to be sure that in the next step we will in fact find the (planted) nearest neighbor $\vc{p}^*$ with reasonable probability. In other words: if $\ip{\vc{p}}{\vc q} \approx \gamma_{max}$, and we therefore cannot easily find any points $\vc{p}' \in \cD \setminus \{\vc{p}^*\}$ anymore which are slightly closer to $\vc{q}$, then we will actually immediately find the exact nearest neighbor $\vc{p}^*$ in the bucket $\B(\vc{p})$ with a certain probability. Note that due to the data set being uniformly random, we may assume that the true (planted) nearest neighbor $\vc{p}^*$ to $\vc{q}$ follows a uniform distribution on the spherical cap $\C_{\vc{q}, \gamma^*}$, where $\gamma^* = \ip{\vc{p}^*}{\vc q}$.

\begin{lemma}[Giant leap success probability] \label{lem:giant}
Let $\vc{p} \in \C_{\vc{q}, \gamma}$ with $\gamma + \frac{1}{d} \geq \gamma_{max}$, and let the exact nearest neighbor to $\vc{q}$ lie at inner product $\gamma^*$ from $\vc{q}$. Then:
\begin{align}
p_g = \Pr_{\vc{p^*} \sim \U(\C_{\vc{q}, \gamma^*})} \Big(\vc{p}^* \in \B(\vc{p})\Big) \geq d^{-O(1)} \frac{W(\alpha, \gamma_{max}, \gamma^*)}{C(\gamma^*)} \, .
\end{align}
\end{lemma}

\begin{proof}
First, observe that the stated probability can be interpreted as: if we draw $\vc{p}^*$ uniformly at random from $\C_{\vc{q}, \gamma^*}$, what is the probability that it is also contained in the spherical cap $\C_{\vc{p}, \alpha}$? This probability corresponds to the ratio between the volume of the corresponding wedge $\W_{\vc{p}, \alpha, \vc{q}, \gamma^*}$ and the volume of the spherical cap $\C_{\vc{q}, \gamma^*}$, as illustrated in the sketch of Figure~\ref{fig:sketch}. Note that the inequality follows from the assumption that $\gamma + \frac{1}{d} \geq \gamma_{max}$; if we replaced this assumption by $\gamma = \gamma_{max}$, then (1) we would not need the leading $d^{-\Theta(1)}$ term, and (2) the inequality would be an equality. 
\end{proof}

As a special case, we observe that if $\alpha$ is sufficiently large (the graph is sufficiently dense), the success probability of Lemma~\ref{lem:giant} is actually non-negligible.

\begin{lemma}[High success probability] \label{lem:high}
For $\alpha \leq \gamma_{max} \cdot \gamma^*$ we have $p_g \geq d^{-O(1)}$.
\end{lemma}

\begin{proof}
This essentially follows from Lemma~\ref{lem:wedge}, which states that $W(\gamma_{max} \gamma^*, \gamma_{max}, \gamma^*) \approx C(\gamma^*)$. If we further decrease the first argument, the result still holds.
\end{proof}

\subsection{Randomization}
\label{sec:proofrand}

The previous analysis is based on doing one iteration of (1) sampling a random node, (2) performing a number of small steps, and (3) hoping that the solution is found as a near neighbor of the last node. For this process we showed that the algorithm succeeds with probability $r \geq 0.98 \cdot p_g$. Ultimately we would like to prove that repeating this process approximately $1/r$ times, starting over at a random new node in graph each time, leads to a constant success probability, even if $r$ is not constant. 

\begin{lemma}[Randomization] Let $r \geq 0.98 \cdot p_g$ denote the success probability of one ``tour'' through the graph, using Algorithm~\ref{alg:query}. Let the number of vertices in the graph encountered in one tour be at most $n^{1 - \eps} r$ with $\eps > 0$. Then for uniformly random data sets, performing $\Theta(1/r)$ random tours through the graph (starting from a random node each time) leads to a constant success probability of finding the solution in at least one of these tours.
\end{lemma}

\begin{proof}
Although starting over at a random node in the graph may not quite lead to independent results, intuitively it is clear that since the number of visited nodes in one tour through the graph is small (insignificant compared to the entire graph), a different walk on the same graph is unlikely to intersect with this previous walk. By making an imaginary adjustment to the query algorithm, where each time a tour fails, we remove all nodes and all neighbors visited in this tour, we guarantee that the next tour occurs in a fresh, new graph, where all nodes in this graph are independent from the nodes in the first tour through the graph. Since the total number of nodes visited by assumption is sublinear in $n$, removing all these previously encountered vertices does not significantly affect the number of remaining nodes in the consecutive subgraphs -- each such imaginary subgraph contains $n(1 - o(1))$ vertices. Also note that since the data set follows a spherically-symmetric distribution, the graph is not biased towards a subset of nodes, which guarantees that removing a small number of vertices does not drastically change the structure of the graph (e.g.\ reduce connectivity). This together shows that by starting over at a new vertex, we can increase the success probability until it becomes constant. 

Finally, observe that if the success probability in one independent trial is $r$, then the probability of success in at least one out of $\Theta(1/r)$ trials is at least $(1 - r)^{\Theta(1/r)} \geq \Theta(1)$; appropriately choosing the leading constant in the number of trials can for instance increase this success probability to $0.99$, without affecting the asymptotics.
\end{proof}

\subsection{Encountered vertices}

Next, let us perform a precise cost analysis of the near neighbor algorithms and the data structure. First, let us look at the number of vertices that Algorithm~\ref{alg:query} encounters in a single tour through the graph, before either finding a local minimum (false positive) or a global minimum (the solution).

\begin{lemma}[Bucket size] \label{lem:bucket}
Let $\alpha \in (0,1)$ satisfy $n \cdot C(\alpha) \gg 1$, and let $\vc{p} \in \cS^{d-1}$ be fixed. Then with overwhelming probability over the randomness of the (remaining) points in the data set $\cD$, the size of the bucket corresponding to $\vc{p}$ satisfies:
\begin{align}
|\B(\vc{p})| = O(n C(\alpha)).
\end{align}
\end{lemma}

\begin{proof}
First, observe that each point in the data set is contained in the bucket $\B(\vc p)$ with probability $C(\alpha)$, and note that these probabilities are independent. The number of vectors in a bucket therefore follows a binomial distribution, and due to the assumption $n \cdot C(\alpha) \gg 1$ we know that we will not encounter a Poisson limit but a Gaussian limit, when the parameters increase. Using standard tail bounds on binomials/Gaussians, we know that the probability that $|\B(\vc p)|$ exceeds its expected value $n C(\alpha)$ by a factor $2$ is very unlikely.
\end{proof}

Using similar arguments, we know that the total number of nodes encountered in one tour through the graph is at most a factor $2 d$ larger (due to each tour consisting of at most $2 d$ steps with high probability).

\begin{lemma}[Encountered vertices in one tour]
Let $\vc{p}_0, \vc{p}_1, \dots, \vc{p}_k$ with $k \leq 2d$ denote the vertices visited in the path on the nearest neighbor graph in a single tour. Then with overwhelming probability over the randomness of the data set:
\begin{align}
\sum_{i=0}^k |\B(\vc{p}_i)| = O(d n C(\alpha)).
\end{align}
\end{lemma}

Doing $\Theta(1/r)$ such tours, to obtain a constant success probability, further increases the number of encountered vectors by a factor $\Theta(1/r)$.

\subsection{Query complexity}

For the query complexity, we state two separate results; one for the costs of doing one greedy tour through the graph, and one for the total costs, when retrying several times in case the algorithm fails to find the solution immediately.

\begin{lemma}[Query complexity, without restarts]
For $\alpha \leq \gamma_{max} \gamma^*$, for large $d$ we will find the nearest neighbor with constant probability without restarts in time at most:
\begin{align}
\texttt{Time} = O(d^2 n C(\alpha)).
\end{align} 
\end{lemma}

\begin{proof}
This follows from the previous analyses of the number of vectors encountered on a single tour, together with the $O(d)$ cost of computing the inner product between a vector in the list and the query vector, to see if this vector is closer to the query than our current near neighbor estimate.
\end{proof}

\begin{theorem}[Query complexity, with restarts]
For arbitrary $\alpha \geq \gamma_{max} \gamma^*$, for large $d$ we will find the nearest neighbor with constant probability with restarts in time at most:
\begin{align}
\texttt{Time} = d^{O(1)} \frac{n C(\alpha) C(\gamma^*)}{W(\alpha, \gamma^*, \gamma_{max})} \, .
\end{align} 
\end{theorem}

\begin{proof}
As explained in Section~\ref{sec:proofrand}, doing many tours through the graph, starting from uniformly random nodes in the graph each time, guarantees that the success probability increases; if the success probability of one tour is $r$, then repeating $\Theta(1/r)$ times leads to a constant success probability. Together with the previous lemma, stating the costs of a single tour, and Lemma~\ref{lem:giant}, stating how many such tours are required to achieve a good success probability, this immediately leads to the stated result.
\end{proof}

\subsection{Space complexity}

Next, let us consider how much memory is required to store the near neighbor graph data structure. To obtain tight bounds on the total number of edges in the graph, note that it does not suffice to only look at (potentially highly dependent) probabilities of edges between vertices; in the extreme case of having a random graph which with probability $C(\alpha)$ is the complete graph and with probability $1 - C(\alpha)$ is the empty graph, knowing the probability $C(\alpha)$ of having an edge between two vertices says little about the probability of having e.g.\ very few or no edges in the resulting graph. In our analysis we therefore need to take into account the geometric interpretation of the graph as well.

\begin{lemma}[Space complexity]
Let $\alpha \in (0,1)$ such that $n C(\alpha) \gg 1$. Then with probability at least $1/\sqrt{n}$, the entire data structure can be stored in space:
\begin{align}
\texttt{Space} = O(n^2 C(\alpha) \log n).
\end{align} 
\end{lemma}

\begin{proof}
In our algorithm, the undirected near neighbor graph $G = (\cD, E)$ consists of $n$ points, with a probability $C(\alpha)$ of an edge $e = \{\vc{p}'$ between two distinct vertices $\vc{p}, \vc{p}' \in \cD$. Due to the geometric interpretation of the points lying on the sphere, and edges representing closeness, we know that if $e, f \in \cD^2$, then $\Pr(e, f \in E) = C(\alpha)^2$ if $e, f$ are disjoint (contain no common endpoints), or share a single endpoint (since the two other points on these edges are independent and uniformly random on the sphere). Only for $e = f$ we have $C(\alpha) = \Pr(e, f \in E) \neq \Pr(e \in E) \Pr(f \in E) = C(\alpha)^2$. Letting $X$ denote the total number of edges in the graph, we know that 
\begin{align}
\Var(X) &= \sum_{e, f \in \binom{\cD}{2}} [\Pr(e, f \in E) - \Pr(e \in E) \Pr(f \in E)] = \sum_{e \in \binom{\cD}{2}} [C(\alpha) - C(\alpha)^2] \\
&= \frac{n (n - 1)}{2} \ C(\alpha) (1 - C(\alpha)).
\end{align}
Using Chebyshev's inequality~\cite[Theorem 4.1.1]{alon08}, we know that $\Pr(|X - \expn(X)| > \lambda \sqrt{\Var(X)}) \leq 1/\lambda^2$ for arbitrary positive $\lambda$. Substituting $\expn(X) = \Theta(n^2 C(\alpha))$ and $\Var(X) = \Theta(n^2 C(\alpha))$, we can for instance take $\lambda = \Theta(n^{1/4})$ (so that $\lambda \sqrt{\Var(X)} < \lambda \sqrt{\Var(X)} n C(\alpha) = o(\expn(X))$) to prove that with probability at most $1/\sqrt{n}$, the number of edges in the graph deviates from its expected value $\Theta(n^2 C(\alpha))$ by at most $\Theta(n^{5/4} \sqrt{C(\alpha)})$.

Finally, observe that each edge essentially represents a pair of integers in $[n]$, which can be represented in $2 \log n$ space. This leads to an extra factor $\log n$ in the space requirement. Due to the assumption that $n C(\alpha) \gg 1$, and the fact that $d$ cannot be too large\footnote{More precisely, using dimension reduction techniques such as the Johnson-Lindenstrauss lemma~\cite{johnson84}, we can always transform the data set beforehand to obtain an equivalent problem in dimension $d = O(\log n \log \log n)$.}, we further know that the above costs are the dominant costs, i.e.\ are higher than the costs of storing the input list of size $O(d n)$ in memory.
\end{proof}

\subsection{Update complexities}

Finally, let us consider the remaining costs for the near neighbor data structure, namely those related to updating the data structure when points are added to/deleted from the data set $\cD$. We state only the most straightforward costs based on Algorithms~\ref{alg:insert}--\ref{alg:delete}; better asymptotics could be obtained using e.g.\ locality-sensitive hashing to index the data structure, but that beats the purpose of using graph-based indexing techniques instead.

\begin{lemma}[Insertion complexity]
Given a data set $\cD$ of $n$ points indexed in an $\alpha$-near neighbor graph, adding a point to this graph (naively) takes time:
\begin{align}
\texttt{Insert} = O(d n).
\end{align}
\end{lemma}

\begin{proof}
The $O(d n)$ insertion complexity is obtained by going through all vertices in the graph, and seeing whether an edge needs to be added to the new vertex. Note that e.g.\ the costs of inserting the new vertex in the buckets of its near neighbors are asymptotically negligible compared to the $O(d n)$ cost for doing $n$ vector comparisons.
\end{proof}

\begin{lemma}[Deletion complexity]
Given a data set $\cD$ of $n$ points indexed in an $\alpha$-near neighbor graph, removing a point from this graph takes time:
\begin{align}
\texttt{Delete} = O(n C(\alpha) \log (n C(\alpha))).
\end{align}
\end{lemma}

\begin{proof}
Let us consider the costs of deleting a vertex, using a data structure where for each vertex, we store a (sorted) list of neighbors (if necessary using a binary search tree). We can easily guery the bucket corresponding to the vertex that needs to be deleted, and for each of the at most $O(n C(\alpha))$ neighbors in this bucket (by Lemma~\ref{lem:bucket}), (1) we remove this neighbor from the bucket of the to-be-deleted vertex, and (2) we remove the to-be-deleted vertex from the bucket corresponding to this neighbor. Note that since each bucket has size $O(n C(\alpha))$, removing the vertex from the neighbor's bucket requires $O(\log(n C(\alpha))$ time using a binary search on this data structure.
\end{proof}


\section{Asymptotics for sparse data sets}
\label{app:sparse}

Recall that Theorem~\ref{thm:main} states that the various costs of this data structure are as follows.
\begin{align}
\texttt{Time} &= d^{O(1)} n C(\alpha) C(\gamma^*) / W(\alpha, \gamma^*, \gamma_{max}), \\
\texttt{Space} &= O(n^2 C(\alpha) \log n), \\
\texttt{Insert} &= O(d n), \\
\texttt{Delete} &= O(n C(\alpha) \log (n C(\alpha))).
\end{align}
Here we will derive asymptotics for large $d$ and $n$ for \textit{sparse} data sets, of size $n = 2^{o(d)}$. As described in Section~\ref{sec:random}, we can represent $n$ in terms of the parameter $\mu$ through the relation $n = 1/C(\mu)$ (up to polynomial factors), with $\mu = o(1)$. We will set $\alpha = \kappa \cdot \mu$ for constant $\kappa \in (0,1)$, so that also $\alpha = o(1)$, and analyze what the asymptotics become as $d, n \to \infty$. Note that $\kappa < 1$ guarantees that $n C(\alpha) \gg 1$.

\begin{lemma}[Relation between $C(\alpha)$ and $C(\mu)$ or $n$] \label{lem:rewn}
Let $\alpha = \kappa \cdot \mu$ with $\mu = o(1)$ and $\kappa \in (0, 1)$ constant. Then, as $d \to \infty$:
\begin{align}
C(\alpha) = C(\mu)^{\kappa^2 + o(1)}.
\end{align}
Therefore, if $n = 1/C(\mu)$, then $C(\alpha) = n^{-\kappa^2 + o(1)}$.
\end{lemma}

\begin{proof}
This is essentially a matter of writing out $C(\alpha)$, using the Taylor expansion of $\ln(1 - x) = -x + O(x^2)$ for small $x$:
\begin{align}
C(\alpha) &= C(\kappa \mu) = d^{O(1)} (1 - \kappa^2 \mu^2)^{d/2} = \exp\left[\frac{d}{2} \ln(1 - \kappa^2 \mu^2) + O(\log d)\right] \\
&= \exp\left[\frac{d}{2} \left(- \kappa^2 \mu^2 + O(\mu^4)\right) + O(\log d)\right] = \exp\left[\frac{d \kappa^2}{2} (- \mu^2) (1 + o(1))\right] \\
&= \exp\left[\frac{d \kappa^2}{2} \ln(1 - \mu^2) (1 + o(1))\right] = C(\mu)^{\kappa^2 + o(1)}.
\end{align}
The relation $C(\alpha) = n^{-\kappa^2 + o(1)}$ then follows immediately as well.
\end{proof}

Using the previous result and $d^{O(1)} = n^{o(1)}$, we obtain the following asymptotics:
\begin{align}
\texttt{Time} &= n^{1 - \kappa^2 + o(1)} C(\gamma^*) / W(\alpha, \gamma^*, \gamma_{max}), \\
\texttt{Space} &= n^{2 - \kappa^2 + o(1)}, \\
\texttt{Insert} &= n^{1 + o(1)}, \\
\texttt{Delete} &= n^{1 - \kappa^2 + o(1)}.
\end{align}
What remains is obtaining asymptotics on $\gamma_{max}$, and expressing $C(\gamma^*)$ and $W(\alpha, \gamma^*, \gamma_{max})$ as powers of $n$, with the exponents being a function of $\kappa$ and $\gamma^*$. The next results follow from this somewhat tedious exercise, starting with $\gamma_{max}$.

\begin{lemma}[Asymptotically reachable $\gamma$] For $\alpha = \kappa \cdot \mu$ with constant $\kappa \in (0, 1)$:
\begin{align}
\gamma_{max} &= \mu \sqrt{1 - \kappa^2} + O(\mu^2). 
\end{align}
\end{lemma}

\begin{proof}
This result is easily obtained by starting from Lemma~\ref{lem:gammax}, substituting $\alpha = \kappa \mu$, and doing a series expansion around $\mu = 0$.
\end{proof}

With these asymptotics for $\gamma_{max}$, we can next derive asymptotics for $W(\alpha, \gamma_{max}, \gamma^*)$ in terms of $\mu, \kappa, \gamma^*$. Here we omit order terms which disappear for large $d$ (and small $\mu$), for clarity of exposition.

\begin{lemma}[Asymptotics for $W(\alpha, \gamma^*, \gamma_{max})$] \label{lem:W}
Let $\alpha = \kappa \cdot \mu$ with $\mu = o(1)$ and $\kappa \in (0, 1)$ constant, and let $\gamma_{max}$ as in Lemma~\ref{lem:gammax}. Then, as $d \to \infty$:
\begin{align}
\hspace{-0.5cm} W(\alpha, \gamma^*, \gamma_{max}) = \begin{cases} \left(\frac{1 -(\gamma^*)^2 - \mu^2 + 2 \gamma^* \mu^2 \sqrt{\kappa^2 \left(1 - \kappa^2\right)}}{1 - \left(1 - \kappa^2\right) \mu^2}\right)^{d/2 + o(d)} & \text{if } 0 < \sqrt{1 - \kappa^2} \mu \leq \min\left\{\frac{\kappa \mu}{\gamma^*}, \frac{\gamma^*}{\kappa \mu}\right\}; \\
(1 - \kappa^2 \cdot \mu^2)^{d/2 + o(d)} & \text{if } \frac{\gamma^*}{\kappa \mu} \leq \sqrt{1 - \kappa^2} \mu < 1; \\ 
(1 - (\gamma^*)^2)^{d/2 + o(d)} & \text{if } \frac{\kappa \mu}{\gamma^*} \leq \sqrt{1 - \kappa^2} \mu < 1. \end{cases}
\end{align}
\end{lemma}

\begin{proof}
This is again a matter of writing out the given expression for $W(\alpha, \gamma^*, \gamma_{max})$ using Lemma~\ref{lem:wedge}, where we substitute $\alpha = \kappa \mu$, we substitute the asymptotics for $\gamma_{max}$ from the previous lemma, and we consider the limiting behavior for small $\mu$ and large $d$.
\end{proof}

The only cases of interest are those where $\gamma^*$ (the target inner product with the planted nearest neighbor) is larger than $\mu$ (the expected maximum inner product over $n$ random data points), which in turn is larger than $\kappa \cdot \mu$ (the query parameter; $\kappa \cdot \mu \approx 1$ would mean having only a small/constant number of neighbors per vertex in the graph). This means that the second case above never applies, since $\frac{\kappa \mu}{\gamma^*} \leq 1 \leq \frac{\gamma^*}{\kappa \mu}$ always holds. 

Next, let us consider the combination of the above expression with $C(\gamma^*)$, using the previous simplification to eliminate the second case:
\begin{align}
\hspace{-0.5cm} \frac{C(\gamma^*)}{W(\alpha, \gamma^*, \gamma_{max})} = \begin{cases} \left(1 + \frac{\kappa^2 + (\gamma^*)^2 \left(1 - \kappa^2\right) - 2 \gamma^* \sqrt{\kappa^2 \left(1 - \kappa^2\right)}}{1 - (\gamma^*)^2} \cdot \mu^2\right)^{d/2 + o(d)} & \text{if } 0 < \gamma^* \leq \sqrt{\frac{\kappa^2}{1 - \kappa^2}}; \\ 
(1 + o(1))^{d/2 + o(d)} & \text{if } \sqrt{\frac{\kappa^2}{1 - \kappa^2}} \leq \gamma^* < 1. \end{cases}
\end{align}
Notice that the second case corresponds exactly to the condition $\alpha \leq \gamma_{max} \gamma^*$ from Lemma~\ref{lem:high}, as expected. This condition can equivalently be rewritten as a condition on $\kappa$ as $\kappa \leq \sqrt{(\gamma^*)^2/ (1 + (\gamma^*)^2)}$, in which case the probability of succeeding in one iteration is non-negligible.

\subsection{Succeeding in one iteration}

In case $\kappa$ is smaller than the above threshold, the success probability of succeeding in one tour through the graph is non-negligible, and the costs of the number of tours through the graph does not contribute to the leading term in the asymptotic complexity. Since all other costs decrease when $\kappa$ increases, the only sensible choice in the regime $\kappa \leq \sqrt{(\gamma^*)^2/ (1 + (\gamma^*)^2)}$ is to let $\kappa$ approach this upper bound, leading to the following costs.

\begin{theorem}[Complexities for succeeding in one iteration]
Using the $\alpha$-NNG with parameter $\alpha = [\sqrt{(\gamma^*)^2/ (1 + (\gamma^*)^2)} - o(1)] \cdot \mu$, with high probability we can solve the sparse near neighbor problem in one iteration with the following complexities:
\begin{align}
\texttt{Time} &= n^{\frac{1}{1 + (\gamma^*)^2} + o(1)}, \\
\texttt{Space} &= n^{\frac{2 + (\gamma^*)^2}{1 + (\gamma^*)^2} + o(1)}, \\
\texttt{Insert} &= n^{1 + o(1)}, \\
\texttt{Delete} &= n^{\frac{1}{1 + (\gamma^*)^2} + o(1)}.
\end{align}
\end{theorem}

Note that for sparse data, the maximum non-planted near neighbor distance to the query is $\sqrt{2} - o(1)$ with high probability. Let us consider the canonical approximate nearest neighbor problem with approximation factor $c$, which in this case corresponds to distances $c \cdot r = \sqrt{2} + o(1)$ and $r = \sqrt{2}/c + o(1)$. By basic geometric arguments on the sphere, an inner product of $\gamma^*$ on the unit sphere translates to a Euclidean distance of $\sqrt{2 (1 - \gamma^*)}$, leading to the relation $\gamma^* = 1 - 1/c^2$; if the planted nearest neighbor has inner product $1 - 1/c^2$ with the query, then its Euclidean distance to the query is a factor $c$ shorter than for other points in the data set.

With this relation between $\gamma^*$ and $c$, we can find asymptotic expressions in terms of $c$ as well. We state two special cases of these asymptotics below, namely for $c \to \infty$ and for $c \approx 1$. These can again be easily proven by taking the previous expressions, substituting the appropriate value of $\gamma^*$, and considering a series expansion around either $c = \infty$ or $c = 1$.

\begin{corollary}[Small/Large-$c$ asymptotics]
Using the $\alpha$-NNG with density parameter $\alpha = [\sqrt{(\gamma^*)^2/ (1 + (\gamma^*)^2)} - o(1)] \cdot \mu$, with high probability we can solve the sparse near neighbor problem in one iteration with the following asymptotic complexities as $d \to \infty$ and $c \to 1^+, \infty$:
\begin{center}
\begin{tabular}{rll} 
\toprule
& \large $\quad (c \to \infty)$ & \large $\quad (c \to 1^+)$ \\ \midrule
\large $\texttt{Time:}$ & \large $n^{\frac{1}{2} + \frac{1}{2c^2} + O(\frac{1}{c^4})}$ & \large $n^{1 - 4 (c - 1)^2 + O((c - 1)^4)}$ \\
\large $\texttt{Space:}$ & \large $n^{\frac{3}{2} + \frac{1}{2c^2} + O(\frac{1}{c^4})}$ & \large $n^{2 - 4 (c - 1)^2 + O((c - 1)^4)}$ \\
\large $\texttt{Insert:}$ & \large $n$ & \large $n$ \\
\large \quad $\texttt{Delete:}$ & \large $n^{\frac{1}{2} + \frac{1}{2c^2} + O(\frac{1}{c^4})}$ \qquad \qquad  & \large $n^{1 - 4 (c - 1)^2 + O((c - 1)^4)}$ \qquad \quad \\ \bottomrule
\end{tabular}
\end{center}
\end{corollary}

Note that the $c \to \infty$ asymptotics above are rather poor; even when the planted nearest neighbor lies very close to the query, the algorithm is bound by a minimum query time of at least $\sqrt{n}$, and a space complexity of at least $n \sqrt{n}$. However, the asymptotics for small $c \approx 1$ match those of Andoni--Laarhoven--Razenshteyn--Waingarten~\cite{andoni17}, with a scaling of $n^{1 - 4(c-1)^2}$ for the query complexity. This already suggests that the larger the approximation factor, the worse the $\alpha$-NNG approach performs in comparison with partition-based approaches.

\subsection{Succeeding in many iterations}

In case $\kappa$ is large, and the graph is rather sparse, many tours through the graph are needed to get a good success probability overall. Taking into account these probabilities, knowing that we are in the ``first case'' due to $\kappa$ being below the threshold, we ultimately obtain the following simplified asymptotics.

\begin{theorem}[Complexities for succeeding in many iterations]
Using the $\alpha$-NNG with $\alpha = \kappa \cdot \mu$ with $\kappa \in (\sqrt{(\gamma^*)^2/ (1 + (\gamma^*)^2)}, 1)$, with high probability we can solve the sparse near neighbor problem in several iterations with the following complexities:
\begin{align}
\texttt{Time} &= n^{\left[\frac{1 - 2 \gamma^* \sqrt{\kappa^2 \left(1 - \kappa^2\right)}}{1 - (\gamma^*)^2} + o(1)\right]}, \\
\texttt{Space} &= n^{2 - \kappa^2 + o(1)}, \\
\texttt{Insert} &= n^{1 + o(1)}, \\
\texttt{Delete} &= n^{1 - \kappa^2 + o(1)}.
\end{align}
\end{theorem}

Similar to~\cite{andoni17}, let us write the space complexity as $n^{1 + \rhos + o(1)}$, and the query complexity as $n^{\rhoq + o(1)}$. Then $\rhos = 1 - \kappa^2$, and we can substitute this into the given asymptotics for $\rhoq$ to obtain a condition on achievable trade-offs between $\rhoq$ and $\rhos$:
\begin{align}
(1 - (\gamma^*)^2) \rhoq + 2 \gamma^* \sqrt{\rhos (1 - \rhos)} \geq 1.
\end{align}
Note that $\alpha = \kappa \mu = \sqrt{1 - \rhos} \mu$ immediately highlights how to construct the appropriate near neighbor graph to achieve such exponents.

Finally, let us also express these results in terms of the approximation factor $c$, using the relation $\gamma^* = 1 - 1/c^2$. Rewriting the above inequality, by substituting this value and multiplying both sides by $c^4$, we obtain the condition stated in the introduction:
\begin{align}
(2c^2 - 1) \rhoq + 2 c^2 (c^2 - 1) \sqrt{\rhos (1 - \rhos)} \geq c^4.
\end{align}
In this case we obtain the same asymptotics for $c \to 1$ and $\rhos \to 0$ as before, but for arbitrary $c$ this indicates which trade-offs are achievable.


\section{Lattice sieving}
\label{app:sieving}

A central hard problem in the study of lattices is the shortest vector problem (SVP): given an integer lattice $\cL = \{\sum_{i=1}^d c_i \vc{b}_i: c_i \in \Z\} \subset \mathbb{R}^d$, find a non-zero lattice vector $\vc{s}$ of minimum Euclidean norm. The security of lattice-based cryptography relies on the hardness of (approximate) SVP, and so algorithms for SVP have long been the subject of study. Currently, the asymptotically fastest methods for SVP rely on a technique called \textit{lattice sieving}, where the algorithm (1) first samples many random, long lattice vectors, and then (2) looks for pairwise combinations of these lattice vectors to form shorter lattice vectors. (Note that if $\vc{v}, \vc{w} \in \cL$ then $\vc{v} \pm \vc{w} \in \cL$.) 

A long line of work, starting from~\cite{laarhoven15crypto}, has studied how state-of-the-art near-neighbor techniques can be used to assist in the search for good pairs of lattice vectors, that can be combined to form shorter vectors. In particular, the hash-based techniques of~\cite{charikar02, andoni14, andoni15cp, becker16lsf} were applied to (heuristic) sieving algorithms in~\cite{laarhoven15crypto, laarhoven15latincrypt, becker16cp, becker16lsf}, to obtain the asymptotic time and space complexities depicted in red in Figure~\ref{fig:sieving}. In this application, the input list is essentially a list of $n = (4/3)^{d/2 + o(d)}$ i.i.d.\ random unit vectors; the target distance on the unit sphere is $1$; and the overall (heuristic) complexity of the algorithm is determined by performing $d^{\Theta(1)} \cdot n$ queries and updates to the data structure. The current best technique dates back to Becker--Ducas--Gama--Laarhoven~\cite{becker16lsf}, which also formed the basis for the generalized time--space trade-offs of Andoni--Laarhoven--Razenshteyn--Waingarten~\cite{andoni17}.

To study what the asymptotics for sieving for SVP may \textit{potentially} be when using graph-based approaches, let us start again with the costs from Theorem~\ref{thm:main}, where for convenience we omit order terms. Here we say ``potentially'' because the update costs (in particular: insertion costs) are currently too high, and for this cost analysis we will assume these costs are the same as the deletion complexity.\footnote{Note that even if this is not the case, and the update costs cannot be reduced, similar techniques/results can still be applied in the context of lattice sieving for solving the closest vector problem with preprocessing (CVPP): given a lattice, preprocess this lattice so that when later given a non-lattice vector, one can quickly determine its closest neighbor in the lattice. Since the trade-offs for sieving for CVPP with near neighbor techniques other than spherical LSF~\cite{laarhoven16cvp2} have not been worked out, we will focus on (potential) SVP complexities, as it allows us to compare several near neighbor techniques in a similar context.}
\begin{align}
\texttt{Time} &= n C(\alpha) C(\gamma^*) / W(\alpha, \gamma^*, \gamma_{max}), \\
\texttt{Space} &= n^2 C(\alpha), \\
\texttt{Update} &= n C(\alpha).
\end{align}
We further have $n = (4/3)^{d/2} \approx 2^{0.2075d}$ and therefore $\mu = \frac{1}{2}$. The target inner product (for Euclidean distance $1$) equals $\gamma^* = \frac{1}{2}$, while using Lemma~\ref{lem:gammax} we obtain:
\begin{align}
\gamma_{max} = \sqrt{\frac{1 - 4 \alpha^2}{5 - 8 \alpha}} \, .
\end{align}
Note that the only remaining variable is $\alpha \in (0, \frac{1}{2})$, and choosing $\alpha$ smaller or larger affects the query time and space complexities directly and indirectly (via $\gamma_{max}$). Computing all complexities for values $\alpha$ in this range (and multiplying the query time complexity by $n$ to obtain the overall time complexity) leads to the blue trade-off curve in Figure~\ref{fig:sieving}.

\subparagraph*{Concrete values.} As can be seen in Figure~\ref{fig:sieving}, the minimum time complexity is obtained at $\alpha \approx 0.4101$. Recall that for $\alpha \to 0$ the graph is very dense, with edges between almost all pairs of vertices; the space complexity then scales as $n^2$, and the time per query becomes superlinear in $n$ due to the small asymptotic success probability of tours, and the high cost of each tour. For $\alpha \to 0$, the space complexity thus becomes $(4/3)^d \approx 2^{0.4150d}$, the time complexity for a single query scales as $(16/11)^{d/2}$, and the overall time complexity for $n$ queries scales as $(64/11)^{d/2} \approx 2^{0.4778d}$. 

The only sensible choices for $\alpha$, leading to true time--space trade-offs, are obtained in the range $\alpha \in (0.4101, 0.50)$. For $\alpha \to \frac{1}{2}$, the space complexity becomes near-linear in the input size, but again the query time complexity becomes superlinear due to the large number of tours required to obtain a reasonable success probability; for $\alpha \to \frac{1}{2}$, we obtain a space complexity of $(4/3)^{d/2} \approx 2^{0.2075d}$, and an overall time complexity of $2^{d/2} = 2^{0.5000d}$. Numerically optimizing for the query complexity led to the value $\alpha \approx 0.4101$, which corresponds to a space complexity of $2^{0.2822d}$ and a time complexity of $2^{0.3274d}$. Note that this range of values for $\alpha$ corresponds to rather sparse graphs; as numerous experiments with near neighbor graphs have shown, it is best to use rather sparse graphs, to obtain the best results.

\subparagraph*{Discussion.} Although these results are only ``potential'' results for SVP (assuming updates can be done as efficiently as in hash-based data structures), we can make some observations. First, similar to the sparse regime, in the context of lattice sieving graph-based near neighbor searching is not better than the best hash-based techniques~\cite{becker16lsf,andoni17}. Perhaps the most practical hash-based library to date however, \FALCONN{}, uses cross-polytope LSH~\cite{andoni15cp}, which is comparable to graph-based trade-offs; for the most part, the cross-polytope trade-off lies slightly above the graph-based trade-off. These asymptotic results seem consistent with benchmarking results of~\cite{annbench1, annbench2, aumueller17}, which suggested that these two techniques are the most competitive for realistic data sets. The well-known hyperplane hashing of Charikar~\cite{charikar02}, studied in the context of sieving in~\cite{laarhoven15crypto}, further does not seem to come close to the asymptotic performance of sieving with graph-based near neighbor techniques; although it is efficient and easy to implement, in high dimensions also graph-based approaches will likely scale better.

\begin{figure}[t]
\begin{center}
\includegraphics[width=12cm]{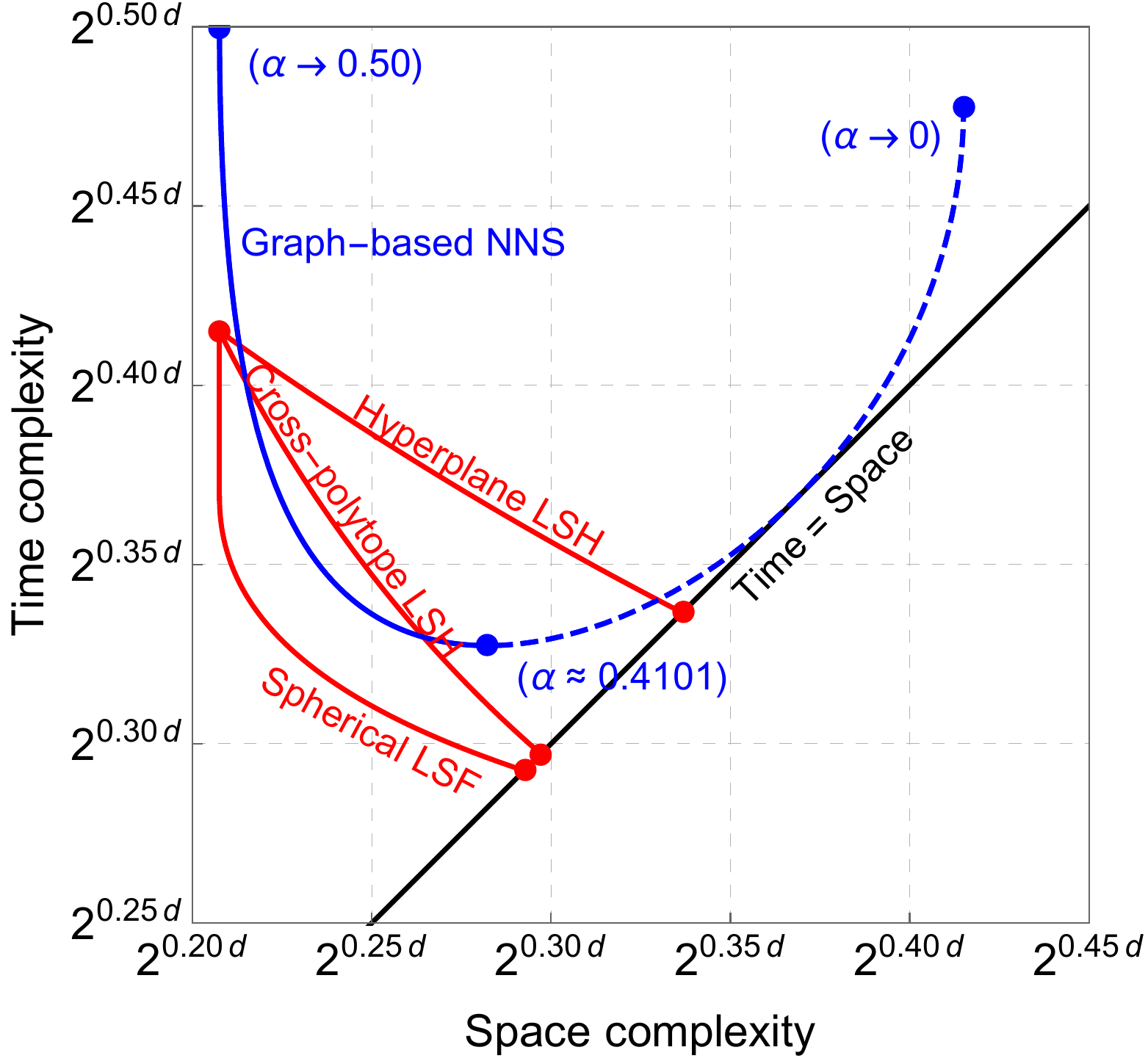} 
\end{center}
\caption{Asymptotic exponents for heuristic lattice sieving methods for solving SVP in dimension $d$, using near neighbor techniques. The red curves correspond to hyperplane LSH~\cite{charikar02, laarhoven15crypto}, cross-polytope/spherical LSH~\cite{andoni14, andoni15cp, laarhoven15latincrypt, becker16cp}, and spherical LSF~\cite{becker16lsf}, and the blue curve indicates \textit{potential} complexities for graph-based near neighbor searching studied in this work.\label{fig:sieving}}
\end{figure}


\newpage

\begin{thebibliography}{10}

\bibitem{alon08}
Noga Alon and Joel~H. Spencer.
\newblock {\em The Probabilistic Method}.
\newblock Wiley-Interscience, 3rd edition edition, 2008.

\bibitem{amsaleg10}
Laurent Amsaleg and Herv\'{e} J\'{e}gou.
\newblock Datasets for approximate nearest neighbor search, 2010.
\newblock URL: \url{http://corpus-texmex.irisa.fr/}.

\bibitem{andoni06}
Alexandr Andoni and Piotr Indyk.
\newblock Near-optimal hashing algorithms for approximate nearest neighbor in
  high dimensions.
\newblock In {\em FOCS}, pages 459--468, 2006.
\newblock URL:
  \url{http://ieeexplore.ieee.org/xpl/articleDetails.jsp?arnumber=4031381},
  \href {http://dx.doi.org/10.1109/FOCS.2006.49}
  {\path{doi:10.1109/FOCS.2006.49}}.

\bibitem{andoni15cp}
Alexandr Andoni, Piotr Indyk, Thijs Laarhoven, Ilya Razenshteyn, and Ludwig
  Schmidt.
\newblock Practical and optimal {LSH} for angular distance.
\newblock In {\em NIPS}, pages 1225--1233, 2015.
\newblock URL:
  \url{https://papers.nips.cc/paper/5893-practical-and-optimal-lsh-for-angular-distance}.

\bibitem{andoni14}
Alexandr Andoni, Piotr Indyk, Huy~L\^{e} Nguy\^{e}n, and Ilya Razenshteyn.
\newblock Beyond locality-sensitive hashing.
\newblock In {\em SODA}, pages 1018--1028, 2014.
\newblock \href {http://dx.doi.org/10.1137/1.9781611973402.76}
  {\path{doi:10.1137/1.9781611973402.76}}.

\bibitem{andoni17}
Alexandr Andoni, Thijs Laarhoven, Ilya Razenshteyn, and Erik Waingarten.
\newblock Optimal hashing-based time-space trade-offs for approximate near
  neighbors.
\newblock In {\em SODA}, pages 47--66, 2017.
\newblock \href {http://dx.doi.org/10.1137/1.9781611974782.4}
  {\path{doi:10.1137/1.9781611974782.4}}.

\bibitem{andoni15}
Alexandr Andoni and Ilya Razenshteyn.
\newblock Optimal data-dependent hashing for approximate near neighbors.
\newblock In {\em STOC}, pages 793--801, 2015.
\newblock \href {http://dx.doi.org/10.1145/2746539.2746553}
  {\path{doi:10.1145/2746539.2746553}}.

\bibitem{andoni17b}
Alexandr Andoni, Ilya Razenshteyn, and Negev~Shekel Nosatzki.
\newblock {LSH} forest: Practical algorithms made theoretical.
\newblock In {\em SODA}, 2017.

\bibitem{arya94}
Sunil Arya, David~M. Mount, Nathan~S. Netanyahu, Ruth Silverman, and Angela~Y.
  Wu.
\newblock An optimal algorithm for approximate nearest neighbor searching in
  fixed dimensions.
\newblock In {\em SODA}, pages 573--582, 1994.
\newblock URL: \url{http://dl.acm.org/citation.cfm?id=314464.314652}.

\bibitem{annbench2}
Martin Aumueller, Erik Bernhardsson, and Alexander Faithfull.
\newblock {ANN} benchmarks -- available online at
  \url{http://sss.projects.itu.dk/ann-benchmarks/}, 2017.
\newblock URL: \url{http://sss.projects.itu.dk/ann-benchmarks/}.

\bibitem{aumueller17}
Martin Aumueller, Erik Bernhardsson, and Alexander Faithfull.
\newblock {ANN}-benchmarks: A benchmarking tool for approximate nearest
  neighbor algorithms.
\newblock In {\em SISAP}, 2017.

\bibitem{bawa05}
Mayank Bawa, Tyson Condie, and Prasanna Ganesan.
\newblock {LSH} forest: self-tuning indexes for similarity search.
\newblock In {\em WWW}, pages 651--660, 2005.

\bibitem{becker16lsf}
Anja Becker, L\'{e}o Ducas, Nicolas Gama, and Thijs Laarhoven.
\newblock New directions in nearest neighbor searching with applications to
  lattice sieving.
\newblock In {\em SODA}, pages 10--24, 2016.
\newblock \href {http://dx.doi.org/10.1137/1.9781611974331.ch2}
  {\path{doi:10.1137/1.9781611974331.ch2}}.

\bibitem{becker16cp}
Anja Becker and Thijs Laarhoven.
\newblock Efficient (ideal) lattice sieving using cross-polytope {LSH}.
\newblock In {\em AFRICACRYPT}, pages 3--23, 2016.
\newblock \href {http://dx.doi.org/10.1007/978-3-319-31517-1_1}
  {\path{doi:10.1007/978-3-319-31517-1_1}}.

\bibitem{annbench1}
Erik Bernhardsson.
\newblock {ANN} benchmarks -- available online at
  \url{https://github.com/erikbern/ann-benchmarks}, 2016.
\newblock URL: \url{https://github.com/erikbern/ann-benchmarks}.

\bibitem{annoy}
Erik Bernhardsson.
\newblock {ANNOY} -- available online at
  \url{https://github.com/spotify/annoy}, 2017.
\newblock URL: \url{https://github.com/spotify/annoy}.

\bibitem{bishop06}
Christopher~M. Bishop.
\newblock {\em Pattern Recognition and Machine Learning (Information Science
  and Statistics)}.
\newblock Springer-Verlag, 2006.

\bibitem{nmslib}
Leonid Boytsov and Bilegsaikhan Naidan.
\newblock {NMSLib} -- available online at
  \url{https://github.com/searchivarius/nmslib}, 2017.
\newblock URL: \url{https://github.com/searchivarius/nmslib}.

\bibitem{brito97}
M.R. Brito, E.L. Chavez, A.J. Quiroz, and J.E. Yukich.
\newblock Connectivity of the mutual $k$-nearest-neighbor graph in clustering
  and outlier detection.
\newblock {\em Statistics \& Probability Letters}, 35(1):33--42, 1997.

\bibitem{charikar02}
Moses~S. Charikar.
\newblock Similarity estimation techniques from rounding algorithms.
\newblock In {\em STOC}, pages 380--388, 2002.
\newblock URL: \url{http://dl.acm.org/citation.cfm?doid=509907.509965}, \href
  {http://dx.doi.org/10.1145/509907.509965} {\path{doi:10.1145/509907.509965}}.

\bibitem{chen09x}
Jie Chen, Haw-ren Fang, and Yousef Saad.
\newblock Fast approximate {$k$NN} graph construction for high dimensional data
  via recursive {Lanczos} bisection.
\newblock {\em Journal of Machine Learning Research}, 10(Sep):1989--2012, 2009.

\bibitem{christiani17}
Tobias Christiani.
\newblock A framework for similarity search with space-time tradeoffs using
  locality-sensitive filtering.
\newblock In {\em SODA}, pages 31--46, 2017.
\newblock \href {http://dx.doi.org/10.1137/1.9781611974782.3}
  {\path{doi:10.1137/1.9781611974782.3}}.

\bibitem{connor10}
Michael Connor and Piyush Kumar.
\newblock Fast construction of $k$-nearest neighbor graphs for point clouds.
\newblock {\em IEEE Transactions on Visualization and Computer Graphics},
  16(4):599--608, 2010.

\bibitem{datar04}
Mayur Datar, Nicole Immorlica, Piotr Indyk, and Vahab~S. Mirrokni.
\newblock Locality-sensitive hashing scheme based on $p$-stable distributions.
\newblock In {\em SOCG}, pages 253--262, 2004.
\newblock \href {http://dx.doi.org/10.1145/997817.997857}
  {\path{doi:10.1145/997817.997857}}.

\bibitem{kgraph}
Wei Dong.
\newblock {KGraph} -- available online at \url{http://kgraph.org/}, 2016.
\newblock URL: \url{http://www.kgraph.org/}.

\bibitem{dong11}
Wei Dong, Moses Charikar, and Kai Li.
\newblock Efficient $k$-nearest neighbor graph construction for generic
  similarity measures.
\newblock In {\em WWW}, pages 577--586, 2011.
\newblock \href {http://dx.doi.org/10.1145/1963405.1963487}
  {\path{doi:10.1145/1963405.1963487}}.

\bibitem{dubiner10}
Moshe Dubiner.
\newblock Bucketing coding and information theory for the statistical
  high-dimensional nearest-neighbor problem.
\newblock {\em IEEE Transactions on Information Theory}, 56(8):4166--4179, Aug
  2010.
\newblock \href {http://dx.doi.org/10.1109/TIT.2010.2050814}
  {\path{doi:10.1109/TIT.2010.2050814}}.

\bibitem{duda00}
Richard~O. Duda, Peter~E. Hart, and David~G. Stork.
\newblock {\em Pattern Classification (2nd Edition)}.
\newblock Wiley, 2000.

\bibitem{eppstein97}
David Eppstein, Michael~S. Paterson, and F.~Frances Yao.
\newblock On nearest-neighbor graphs.
\newblock {\em Discrete \& Computational Geometry}, 17(3):263--282, 1997.

\bibitem{hajebi11}
Kiana Hajebi, Yasin Abbasi-Yadkori, Hossein Shahbazi, and Hong Zhang.
\newblock Fast approximate nearest-neighbor search with $k$-nearest neighbor
  graph.
\newblock In {\em IJCAI}, volume~22, pages 1312--1317, 2011.

\bibitem{indyk98}
Piotr Indyk and Rajeev Motwani.
\newblock Approximate nearest neighbors: Towards removing the curse of
  dimensionality.
\newblock In {\em STOC}, pages 604--613, 1998.
\newblock \href {http://dx.doi.org/10.1145/276698.276876}
  {\path{doi:10.1145/276698.276876}}.

\bibitem{johnson84}
William~B. Johnson and Joram Lindenstrauss.
\newblock Extensions of {L}ipschitz mappings into a {H}ilbert space.
\newblock {\em Contemporary Mathematics}, 26(1):189--206, 1984.
\newblock \href {http://dx.doi.org/10.1090/conm/026/737400}
  {\path{doi:10.1090/conm/026/737400}}.

\bibitem{kennedy17}
Christopher Kennedy and Rachel Ward.
\newblock Fast cross-polytope locality-sensitive hashing.
\newblock In {\em ITCS}, 2017.
\newblock URL: \url{https://arxiv.org/abs/1602.06922}.

\bibitem{laarhoven15crypto}
Thijs Laarhoven.
\newblock Sieving for shortest vectors in lattices using angular
  locality-sensitive hashing.
\newblock In {\em CRYPTO}, pages 3--22, 2015.
\newblock \href {http://dx.doi.org/10.1007/978-3-662-47989-6_1}
  {\path{doi:10.1007/978-3-662-47989-6_1}}.

\bibitem{laarhoven15nns}
Thijs Laarhoven.
\newblock Tradeoffs for nearest neighbors on the sphere.
\newblock {\em arXiv}, pages 1--16, 2015.
\newblock URL: \url{https://arxiv.org/abs/1511.07527}.

\bibitem{laarhoven16cvp2}
Thijs Laarhoven.
\newblock Finding closest lattice vectors using approximate {V}oronoi cells.
\newblock {\em Cryptology ePrint Archive, Report 2016/888}, 2016.
\newblock URL: \url{https://eprint.iacr.org/2016/888}.

\bibitem{laarhoven17hypercube}
Thijs Laarhoven.
\newblock Hypercube {LSH} for approximate near neighbors.
\newblock In {\em MFCS}, 2017.

\bibitem{laarhoven15latincrypt}
Thijs Laarhoven and Benne de~Weger.
\newblock Faster sieving for shortest lattice vectors using spherical
  locality-sensitive hashing.
\newblock In {\em LATINCRYPT}, pages 101--118, 2015.
\newblock \href {http://dx.doi.org/10.1007/978-3-319-22174-8_6}
  {\path{doi:10.1007/978-3-319-22174-8_6}}.

\bibitem{malkov16}
Y.A. Malkov and D.A. Yashunin.
\newblock Efficient and robust approximate nearest neighbor search using
  {Hierarchical Navigable Small World} graphs.
\newblock {\em arXiv:1603.09320}, 2016.

\bibitem{malkov14}
Yury Malkov, Alexander Ponomarenko, Andrey Logvinov, and Vladimir Krylov.
\newblock Approximate nearest neighbor algorithm based on navigable small world
  graphs.
\newblock {\em Information Systems}, 45:61--68, 2014.

\bibitem{may15}
Alexander May and Ilya Ozerov.
\newblock On computing nearest neighbors with applications to decoding of
  binary linear codes.
\newblock In {\em EUROCRYPT}, pages 203--228, 2015.
\newblock \href {http://dx.doi.org/10.1007/978-3-662-46800-5_9}
  {\path{doi:10.1007/978-3-662-46800-5_9}}.

\bibitem{micciancio10b}
Daniele Micciancio and Panagiotis Voulgaris.
\newblock Faster exponential time algorithms for the shortest vector problem.
\newblock In {\em SODA}, pages 1468--1480, 2010.
\newblock URL: \url{http://dl.acm.org/citation.cfm?id=1873720}.

\bibitem{miller97}
Gary~L. Miller, Shang-Hua Teng, William Thurston, and Stephen~A. Vavasis.
\newblock Separators for sphere-packings and nearest neighbor graphs.
\newblock {\em Journal of the ACM}, 44(1):1--29, 1997.

\bibitem{flann}
Marius Muja and David~G. Lowe.
\newblock {FLANN} -- available online at
  \url{https://www.cs.ubc.ca/research/flann/}, 2013.

\bibitem{pennington14}
Jeffrey Pennington, Richard Socher, and Christopher~D. Manning.
\newblock Glove: Global vectors for word representation.
\newblock In {\em Empirical Methods in Natural Language Processing (EMNLP)},
  pages 1532--1543, 2014.
\newblock URL: \url{http://www.aclweb.org/anthology/D14-1162}.

\bibitem{plaku07}
Erion Plaku and Lydia~E. Kavraki.
\newblock Distributed computation of the $k$nn graph for large high-dimensional
  point sets.
\newblock {\em Journal of Parallel and Distributed Computing}, 67(3):346--359,
  2007.

\bibitem{ponomarenko11}
Alexander Ponomarenko, Yury Malkov, Andrey Logvinov, and Vladimir Krylov.
\newblock Approximate nearest neighbor search small world approach.
\newblock In {\em ICTA}, 2011.

\bibitem{falconn}
Ilya Razenshteyn and Ludwig Schmidt.
\newblock {FALCONN} -- available online at \url{https://falconn-lib.org/},
  2016.
\newblock URL: \url{https://falconn-lib.org/}.

\bibitem{shakhnarovich05}
Gregory Shakhnarovich, Trevor Darrell, and Piotr Indyk.
\newblock {\em Nearest-Neighbor Methods in Learning and Vision: Theory and
  Practice}.
\newblock MIT Press, 2005.
\newblock URL: \url{http://ttic.uchicago.edu/~gregory/annbook/book.html}.

\bibitem{terasawa07}
Kengo Terasawa and Yuzuru Tanaka.
\newblock Spherical {LSH} for approximate nearest neighbor search on unit
  hypersphere.
\newblock In {\em WADS}, pages 27--38, 2007.
\newblock \href {http://dx.doi.org/10.1007/978-3-540-73951-7_4}
  {\path{doi:10.1007/978-3-540-73951-7_4}}.

\bibitem{wang12}
Jing Wang, Jingdong Wang, Gang Zeng, Zhuowen Tu, Rui Gan, and Shipeng Li.
\newblock Scalable $k$-nn graph construction for visual descriptors.
\newblock In {\em CVPR}, pages 1106--1113, 2012.

\end{thebibliography}
\end{document}